\documentclass[a4paper,UKenglish,cleveref, autoref, thm-restate]{lipics-v2021}

\pdfoutput=1 
\hideLIPIcs  

\title{On Higher-Order Probabilistic Verification\\ via the Weighted Relational Model of Linear Logic} 

\author{Ugo Dal Lago}{University of Bologna, Italy, \and INRIA Sophia Antipolis, France }{ugo.dallago@unibo.it}{https://orcid.org/0000-0001-9200-070X}{}
\author{Guido Fiorillo}{Université Claude Bernard Lyon 1, France }{guido.fiorillo@ens-lyon.fr}{}{}
\author{Paolo Pistone}{Université Claude Bernard Lyon 1, France }{paolo.pistone@ens-lyon.fr}{https://orcid.org/0000-0003-4250-9051}{}

\authorrunning{U. Dal Lago, G. Fiorillo, P. Pistone} 

\Copyright{Ugo Dal Lago, Guido Fiorillo, Paolo Pistone} 



\category{} 

\relatedversion{} 

\nolinenumbers 

\EventEditors{John Q. Open and Joan R. Access}
\EventNoEds{2}
\EventLongTitle{42nd Conference on Very Important Topics (CVIT 2016)}
\EventShortTitle{CVIT 2016}
\EventAcronym{CVIT}
\EventYear{2016}
\EventDate{December 24--27, 2016}
\EventLocation{Little Whinging, United Kingdom}
\EventLogo{}
\SeriesVolume{42}
\ArticleNo{23}
\usepackage{algorithmic}
\usepackage{graphicx}
\usepackage{textcomp}
\usepackage{xcolor}
\usepackage{stmaryrd}
\usepackage{hyperref}

\usepackage{adjustbox}
\usepackage{proof}

\usepackage{bussproofs}\EnableBpAbbreviations
\usepackage{subcaption}
\usepackage{listings}

\newenvironment{varitemize}
{
	\begin{list}{\labelitemi}
		{\setlength{\itemsep}{0pt}
			\setlength{\topsep}{0pt}
			\setlength{\parsep}{0pt}
			\setlength{\partopsep}{0pt}
			\setlength{\leftmargin}{15pt}
			\setlength{\rightmargin}{0pt}
			\setlength{\itemindent}{0pt}
			\setlength{\labelsep}{5pt}
			\setlength{\labelwidth}{10pt}
	}}
	{
	\end{list} 
}

\newcounter{numberone}

\newenvironment{varenumerate}
{
	\begin{list}{\arabic{numberone}.}
		{
			\usecounter{numberone}
			\setlength{\itemsep}{0pt}
			\setlength{\topsep}{0pt}
			\setlength{\parsep}{0pt}
			\setlength{\partopsep}{0pt}
			\setlength{\leftmargin}{15pt}
			\setlength{\rightmargin}{0pt}
			\setlength{\itemindent}{0pt}
			\setlength{\labelsep}{5pt}
			\setlength{\labelwidth}{15pt}
	}}
	{
	\end{list} 
}

\newcommand{\modd}[1]{\llbracket#1\rrbracket }

\newcommand{\fps}[2]{{#1}\{\!\!\{#2\}\!\!\}}
\newcommand{\fpp}[2]{{#1}\{#2\}}

\newcommand{\model}[1]{\modd{#1}}

\newcommand{\C}[1]{\mathcal{#1}}
\newcommand{\BB}[1]{\mathbb{#1}}

\newcommand{\Etor}{\mathbf{ETR}}

\newcommand{\LY}{\lambda Y}
\newcommand{\PLY}{P\lambda Y}
\newcommand{\Gfin}{\text{PBHORS}^{<\infty}}
\newcommand{\Ginf}{\text{PBHORS}^{\infty}}
\newcommand{\AFF}[1]{#1_{\mathrm{aff}}}
\newcommand{\Aff}{{\mathrm{Aff}}}
\newcommand{\Com}{\mathbf{Com}}

\newcommand{\nonterm}{\C N}
\newcommand{\tri}{\triangleleft}

\newcommand{\sat}[1]{{#1}_{\mathrm{st}}}

\newcommand{\N}{\BB N}
\newcommand{\Q}{\BB Q}
\newcommand{\R}{\BB R}
\newcommand{\Qrel}[1]{#1\mathbf{Rel}}
\newcommand{\Qrelkleisli}[1]{#1\mathbf{Rel_!}}

\newcommand {\redbigp}[1]{\xrightarrow{#1}}

\newcommand{\bydef}{\: := \:}

\newcommand{\1}{o}
\newcommand{\phors}{{(\mathcal{N}, \mathcal{T}, \mathcal{D}, S)}}
\newcommand{\gphors}{\mathcal{G}}
\newcommand{\Pterm}[1]{\mathbb{P}({#1} \downarrow)}
\newcommand{\ExpTime}[1]{\mathbb{E}({#1} \downarrow)}

\newcommand{\fix}{\mathsf{fix}\:}
\newcommand{\Rsemiring}{\mathcal{R}}
\newcommand{\Qsemiring}{\mathcal{Q}}
\newcommand{\Rinf}{\mathbb{R}^{+\infty}_{\geq 0}}
\newcommand{\Qpos}{\mathbb{Q}_{\geq 0}}
\newcommand{\supp}{\mathsf{supp} \:}

\newcommand{\exptower}[2]{\exp_{#1}({#2})}

\providecommand{\dotdiv}{
  \mathbin{
    \vphantom{+}
    \text{
      \mathsurround=0pt 
      \ooalign{
        \noalign{\kern-.35ex}
        \hidewidth$\smash{\cdot}$\hidewidth\cr 
        \noalign{\kern.35ex}
        $-$\cr 
      }%
    }%
  }%
}

\begin{document}

\maketitle

\begin{abstract}

The problem of determining whether a probabilistic program terminates almost surely (i.e.~with probability one) is undecidable, and actually $\Pi^0_2$-complete. 
For this reason, a growing literature has explored classes of programs for which this and related problems can be shown (semi-)decidable.
In this work we consider the termination problem for the language of Probabilistic Higher-Order Recursion Schemes (PHORS). 
Using the weighted relational semantics of linear logic, we translate this problem into the computation 
of suitable generating functions associated with the program interpreted. 
This way, we establish the decidability of almost sure termination for a class of programs that extends Li et al.'s affine PHORS via a type discipline with bounded exponentials.
To achieve this, we show that the generating functions for such programs are always algebraic, that is, solutions of polynomial equations, yielding an effective method to answer the termination problem.
\end{abstract}

\section{Introduction}


\paragraph*{Probabilistic Higher-Order Verification}



Model checking is one of the most successful techniques for the verification of programs and systems \cite{APT1986307, 10.1145/1592761.1592781}, and it has found applications in a wide range of areas, leading to the development of widely used tools such as SPIN \cite{10.1109/32.588521} and PRISM \cite{10.5555/647810.738106}. While in its original formulation model checking focused on finite-state systems modeled by Kripke structures, its scope of application has gradually expanded over time, in directions that are often very different from one another. In this paper, we are particularly interested in two such directions.

The first concerns the model checking of systems whose evolution is intrinsically \emph{probabilistic}. Along this line of research, the literature has produced a rich body of results, both positive and negative, regarding the decidability and tractability of the underlying verification problem (see \cite{Katoen} for a recent survey). It is worth noting that when the property to be verified is reachability or termination (appropriately generalized to a probabilistic setting), verification may remain decidable even beyond the finite-state case, as in, for example, recursive Markov chains \cite{10.1145/2159531.2159534} or pushdown automata \cite{10.5555/1018438.1021838}.

But there is also another generalization of model checking, likewise obtained by considering a broader class of systems than those traditionally studied, namely so-called \emph{higher-order model checking}. It has been known for about twenty years that this problem is decidable when the underlying specification is an arbitrary MSO formula and the program is an \emph{higher-order recursion scheme} (HORS, \cite{10.1109/LICS.2006.38, DBLP:conf/lics/Ong15}), a language equivalent to the simply-typed $\lambda$-calculus enriched with a fixpoint operator. Problems such as termination or reachability are therefore decidable for a broad class of programs encompassing higher-order types and recursion. These positive results have nonetheless revealed
a certain degree of fragility: relatively minor extensions suffice to render the problem undecidable.

What, then, happens when one considers model-checking problems for probabilistic \emph{and} higher-order programs? Recently, Kobayashi et al.~\cite{DBLP:journals/lmcs/KobayashiLG20} have introduced \emph{probabilistic higher-order recursive schemes} (PHORS), a probabilistic variant of HORS, which can be seen as the natural randomized variation on the theme of HORS. Unfortunately, verifying natural probabilistic counterparts of reachability and termination properties is in general a much more difficult task and is undecidable for PHORS, contrary to what happens for HORS and recursive Markov chains. For example, in a probabilistic setting, a natural property is \emph{almost sure termination} (AST), that is, termination with probability 1. Now, while termination is the quintessential \emph{semi-decidable} property for a Turing-complete language (and is even \emph{decidable} for HORS),  AST is $\Pi^0_2$-complete in Turing-complete probabilistic languages \cite{DBLP:journals/acta/KaminskiKM19}, and its decidability fails for PHORS, already at order 2 \cite{DBLP:journals/lmcs/KobayashiLG20}. For these reasons, the recent literature has focused on capturing sub-classes of PHORS for which AST and, possibly, the related \emph{positive almost sure termination} (PAST) -- the finiteness of the \emph{expected} number of steps to termination -- could be shown decidable, and thus possibly amenable to verification and model-checking techniques. Notably, while \cite{DBLP:journals/lmcs/KobayashiLG20} established the decidability of AST for order-1 PHORS, Li et al.~\cite{DBLP:conf/lics/LiMO22} established the decidability of both AST and PAST for the \emph{affine} PHORS (styled PAHORS), i.e.~the recursive programs which are allowed to use each of their inputs \emph{at most} once during their reductions.

Despite the aforementioned results, the nature of the problem remains poorly understood, and a precise assessment of its computational difficulty is still lacking. For instance, it is currently unknown whether the decidability problem for PHORS remains $\Pi^0_2$-complete, or whether it can, while being undecidable, be placed lower in the arithmetical hierarchy. Similarly, while different methods (like intersection types, game semantics or automata theory) have been applied to achieve decidability, a general approach encompassing all these results has not yet been identified. We therefore believe that it is important to investigate the nature of this problem in greater depth, which is the primary goal of this paper.

%
%
%

\paragraph*{Probabilistic Termination via Generating Functions}

Estimating the probability of termination of higher-order programs is hard, because it demands to \emph{count} probabilities over a typically \emph{infinite} set of reductions, a problem that remains difficult even when these reductions follow some well-defined recursive pattern. 

In fact, similar counting problems are common in the field of combinatorics, in which they are addressed by studying the algebraic and analytic properties of the corresponding \emph{generating functions}, i.e.~of the formal power series that are naturally associated with the sequence of numbers that one wishes to compute or estimate.   
A famous example is provided by the well-known \emph{Catalan numbers} $C_n=\frac{1}{n+1}\binom{2n}{n}$, where $C_n$ is the number of labeled binary trees with $n$ nodes. 
This sequence induces the generating function $c(x)=\sum_{n=0}^{\infty}C_n x^n$, which can be characterized as the solution of a simple algebraic equation, namely $xc(x)^2-c(x)+1=0$. 
From a polynomial equation like this it is possible to extract informatios about the generating function and the corresponding sequence that it would be otherwise very hard to compute directly.
In our example, it allows us to deduce the closed expression $c(x)=\frac{1-\sqrt{1-4x}}{2x}$, hence providing a way to compute the values of $c(x)$ effectively.

%

%

Starting from Chomsky and Sch\"utzenberger's seminal work \cite{CHOMSKY1963118}, generating functions have been widely applied in formal language theory. For example, it is well-known that for all regular languages $L$, the associated generating function $L(x)=\sum_{n=0 }^{\infty}L_nx^n$, where $L_n$ counts the words of $L$ of length $n$, is \emph{rational}, that is, it can be written as a fraction $L(x)=p(x)/q(x)$ of two polynomials.
%
%
When $L$ is context-free, instead, $L(x)$ is \emph{algebraic}, i.e.~, it is the solution, like $c(x)$ above, of some polynomial equation $p(x,L(x))=0$. 

%

Now, suppose $\gphors$ is a PHORS, i.e.~an higher-order recursive probabilistic program, and that $\gphors$ makes, at each of its reduction steps, an unbiased choice between two different alternatives. 
Letting $\gphors_n$ be the number of \emph{distinct} terminating reductions of $\gphors$ with exactly $n$ steps, the power series $a_\gphors(z)=\sum_{n=0}^{\infty}\frac{a_n}{2^n}z^n$ precisely captures the probabilistic behavior of $\gphors$. Notably, $a_\gphors(1)$ computes its probability of termination, and, as we'll see, the derivative $a_{\gphors}'(1)$ captures the \emph{expected} number of reduction steps to termination.
 
A natural question is thus: can we find a way to relate the complexity of the termination problem for a PHORS $\gphors$ with the algebraic or analytic properties of its generating function $a_\gphors(z)$?

\paragraph*{Generating Functions via Linear Logic}

Generating functions have been applied in the literature to the study of
both formal grammars \cite{DBLP:books/sp/KuichS86} and 
 \emph{first-order} programming languages (also in the probabilistic case, \cite{Klinkenberg2021,Klinkenberg2024}).  However, extending this natural and powerful approach to functional programs is challenging, as it is \emph{prima facie} not obvious how to associate \emph{higher-order} functions  with sequences of real numbers, as required to induce a generating function.

In this work we provide a solution to this problem, by presenting the first combinatorial study of the termination problem for {higher-order} recursive programs like PHORS.
To do this, we exploit in an essential way two ideas coming from the toolbox of linear logic.
The first fundamental ingredient is the \emph{weighted relational model} \cite{DBLP:conf/lics/LairdMMP13}, which associates proofs in linear logic, as well as PCF programs, with families of \emph{formal power series} with coefficients in a continuous semiring.
We show how this model can be used to associate each PHORS with a corresponding generating function $a_\gphors(z)$ in an elegant and compositional way.

As mentioned above, linearity has been recognized as a key factor to ensure the decidability of the AST problem for PHORS \cite{DBLP:conf/lics/LiMO22}.
The viewpoint of generating functions provides a \emph{novel} way to look at this result:
using our interpretation via the weighted relational model, we will show that for all affine PHORS, while the corresponding language might \emph{not} be context-free, the corresponding generating function is, in fact, always algebraic.
 Actually, the combinatorial viewpoint allows us to go beyond the linear/affine restriction and consider even \emph{non-linear} PHORS typable via the discipline of \emph{graded} or \emph{bounded exponentials} \cite{GIRARD19921, 10.1007/978-3-642-02273-9_8, 10.1007/978-3-642-54833-8_19, 10.1007/978-3-642-54833-8_18}: this discipline allows one to define types of the form $!_n A \multimap B$, intuitively expressing functions from $A$ to $B$ using their inputs \emph{at most} $n$ times. 
 
 Hence, on the one hand, in this work we capture a \emph{larger} class of programs than \cite{DBLP:conf/lics/LiMO22}, including PHORS admitting both bounded and unbounded forms of duplication, for which the AST and PAST problems remains decidable.
On the other, our results suggest the \emph{robustness} of the class of PHORS individuated by \cite{DBLP:conf/lics/LiMO22}: all PHORS typable in our discipline are shown 
 {equivalent} (in terms of both probabilistic termination and the corresponding language) to some PAHORS, even though reconstructing the latter effectively may lead to an exponential blow up in size.

\paragraph*{Contributions}
Our contributions can be resumed as follows:
\begin{varitemize}

\item First, we show that the weighted relational semantics naturally associates each PHORS with a generating function capturing its probability of termination and expected number of steps to termination, whose coefficients are implicitly defined via a system of (countably many) recursive equations. 

\item Then, we introduce the \emph{finitely bounded PHORS}, $\Gfin$, extending the PAHORS of \cite{DBLP:conf/lics/LiMO22} with finitely bounded non-linearity, and show that the corresponding generating functions are algebraic, leading to the decidability of AST and PAST.

\item Finally, we introduce an even larger class $\Ginf$, extending the $\Gfin$ with a restricted use of \emph{unbounded} non-linearity, still preserving algebraicity and decidability. Notably, the $\Ginf$ can be \emph{composed}, thus enabling a modular study of their generating functions.
%
%
%

\end{varitemize}

\section{From PHORS to Generating Functions, via Linear Logic}

In this section we provide an overview of our approach to probabilistic termination via algebraic generating functions.

\paragraph*{The Generating Function of a PHORS}

Consider a program $\gphors$ defined by the following equations:
\begin{equation}\label{eq:introphors1}
\begin{aligned}
Fx&= x\oplus_{\frac{1}{2}} Fx\\
S& =F \  e
\end{aligned}
\end{equation}
where $F:o\to o$. 
This is a very simple example of a order-1 PHORS: the upper case letters $F,S$ are called \emph{non-terminal} symbols, the execution of the program starts from the \emph{source} non-terminal $S:o$ by applying instances of the equations, read from left to right, as well as probabilistic choices, terminating when the unit constant $e:o$ is, eventually, produced. 
If we ignore probabilities, we can consider the \emph{branch language} $\C L(\gphors)$ computed by the PHORS $\gphors$ as follows:
 replace each choice $t\oplus_{\frac{1}{2}}u$ with a non-deterministic operation $ct+cu$ marked by an order-1 function symbol $c:o\to o$; this leads to consider the following (non-deterministic) HORS $\gphors'$,  where $F':( o\to o)\to \tau$:
\begin{equation}\label{eq:introphors2}
\begin{aligned}
F'zx&=  zx + F'zx\\
S& =F' \ c \ e
\end{aligned}
\end{equation} 
The non-deterministic execution of $\gphors'$ produces an infinite tree whose finite branches form the language $\C L(\gphors)=\{c^{n}e\mid n\in \N\}$.

Our goal is to extract a generating function from programs like \eqref{eq:introphors1}, i.e.~a formal power series $a_\gphors(z)=\sum_{i=0}^\infty a_iz^i$ capturing the probabilistic termination of $\gphors$. Recall that a power series induces an \emph{analytic function} defined over some open subset $U\subset \mathbb C$ of the field of complex numbers, which, when this set is not empty, can thus be studied with methods coming from both algebra and complex analysis. For instance, the formal power series $a(x)=\sum_{n=0}^\infty \frac{1}{2^n} x^n$ is equal to the analytic function $a(x)=\frac{1}{1-\frac{x}{2}}$ over the open set of complex numbers of absolute value strictly less than $ 2$.

As anticipated, our fundamental ingredient to extract generating functions from PHORS is the weighted relational model of linear logic. This semantics is often styled \emph{quantitative}, as it provides a \emph{precise count} of the number of times that a program may use each of its inputs during any of its reductions.
If we consider the HORS $\gphors'$, we see that, for any $i\in \N$, the term $F'zx$ has a unique way of terminating using $z$ exactly $i$ times. Formally, each such reduction yields, in the semantics, a monomial of the form $z^{i}x$. Notice that the parameter $z$ precisely counts the number of times that a choice is made.
By considering all possible reductions, the corresponding monomials yield then the power series below:
\[
a_{F'}(z,x)=\sum_{i>0}^{\infty}z^i x
\]
To get back to the original PHORS $\gphors$ we can simply add probabilities, i.e.~
interpret $x\oplus_{\frac{1}{2}}y$ as $z(\frac{1}{2}x+\frac{1}{2}y)$, yielding the rational power series
%
\[
a_{F}(z,x)=\sum_{i>0}^{\infty}\frac{1}{2^i}z^i x=\frac{x}{1-\frac{z}{2}}-x
\]
saying that the probability that $Fx$ terminates using $x$ once and making $i>0$ choices (i.e.~using $z$ $i$ times) is $\frac{1}{2^i}$. 
As $e$ trivially terminates with probability 1, we can then associate the source symbol $S$ with the generating function
$a_{S}(z):=a_{F}(z,1)=\sum_{i>0}^{\infty}\frac{1}{2^i}z^i=\frac{1}{1-z/2}-1$,
which describes the reductions with $i$ choices, and finally deduce that $\gphors$ terminates with probability $a_{S}(1)=1$, i.e.~almost surely.

Using the ideas just sketched, as we show in Section 4, any PHORS $\gphors$ can be associated with a generating function 
$
a_{\gphors}(z):=a_{S}(z)=\sum_{n=0}^\infty a_n z^n
$
such that $a_n$ designates the probability that $\gphors$ terminates after \emph{exactly} $n$ reduction steps. Observe that the probability of termination of $\gphors$ is given by $a_{\gphors}(1)$. Moreover, $a'_{\gphors}(1)$, where
$a'_{\gphors}(z)=\sum_{n=0}^{\infty}na_n z^{n+1}$ indicates 
the \emph{derivative} of $a_{\gphors}(z)$, precisely captures the expected number of reduction steps to termination.

It is worth observing at this point that, in our setting, the generating function $a_\gphors(z)$, while characterizing the termination problem for the PHORS $\gphors$,  \emph{does not} characterize the corresponding branch language: for instance, two PHORS respectively inducing the non context-free language $\C L_1=\{0^n1^n0^n\mid n\in \N\}$ and the context-free language 
$\C L_2=\{0^n1^{2n}\mid n\in \N\}$ would \emph{both} induce a rational generating function of the form $a(z)=\sum_{n=0}^{\infty}\frac{1}{2^{3n}}z^{3n}=\frac{1}{1-\left({z}/{2}\right)^3}$, as the latter only counts the number of reduction steps to termination.

%
%
%
%
%
%
%
%
%
%

\paragraph*{Algebraic Generating Functions via Bounded Exponentials}

Consider now a slightly more complex PHORS $\gphors$, defined by:
\begin{equation}\label{eq:phors1}
\begin{aligned}
Hfx&=( H(A\circ f)x \oplus_{a}
H(B\circ f)x )\oplus_{a}f(fx)\\
Ax&= x\oplus_{b} \Omega\\
Bx&=x\oplus_{c} \Omega\\
S&=HIe,
\end{aligned}
\end{equation}
where $H:(o\to o)\to (o\to o)$, $A,B:(o\to o)$ and $S:o$, $\Omega$ stands for a diverging term, 
and $a,b,c,d$ stand for rational biases for the probabilistic choice operators. Observe that $H$ is not linear, as it may use the variable $f$ twice.
%
Using distinct constants $a,b,c:o\multimap o$ for each choice symbol we can see that the branch language of $\gphors$, 
%
$$
\mathcal L(\gphors)=\{ a^{2|w|+2}ww\mid w\in \{b,c\}^*\}.
$$
is not context-free, as it contains an arbitrary word repeated twice.

In the weighted relational semantics, the non-terminal $H$ yields, as before, a power series $a_{H}(z,f,x)=\sum_{i}^{\infty}H_{i}(z) f^{i} x$, where
$H_{i}(z)=\sum_j H_{ij}z^j$ counts the probability that $Hfx$ terminates using $f$ exactly $i$ times, $x$ once, but running through an arbitrary number of choices.
By translating \eqref{eq:phors1} into equations between the corresponding power series, we are led to: 
\begin{equation}\label{eq:poly}
\begin{aligned}
H_{2}(z)&=
\alpha H_{2}(z)z^2  +\beta z, \qquad\qquad  H_{i}(z)=0 \quad (i\neq 2),
\end{aligned}
\end{equation}
where $\alpha=(a^2b^2+a(1-a)c^2)$ and $\beta=1-a$.
The power series 
$a_{H}(z,f,x)=H_2(z)f^2x$ can then be computed by finding the minimal non-negative solution of the \emph{polynomial equation}
\[
p(z,H_2(z))=0,
\]
where $p(z,w)= (\alpha z^2-1)w+\beta z$, yielding 
$a_{H}(z,f,x)= \frac{\beta z}{1-\alpha z^2}f^2x$. 
This yields the termination probability $a_{\gphors}(1)=
a_{H}(1,1,1)=
\frac{\beta }{1-\alpha }$.

Unfortunately, termination probabilities cannot always be computed algebraically, as above.
For instance, it is well-known that the generating functions of order-2 PHORS (which correspond to the so called \emph{indexed grammars}) may fail to be algebraic \cite{DBLP:journals/ita/AdamsFM13}. 
In our semantics, the generating function $a_{\gphors}(z)$ of a PHORS $\gphors$ will be determined by a (countably) infinite family of equations, which we might not be able to solve algebraically. 
 
The main theme of this work is then to find conditions under which the generating function $a_{\gphors}(z)$ may be found as the minimal solution of some \emph{finite system of polynomial equations}. This will have two applications: in the first place, thanks to the characterization in term of algebraicity \textit{and} minimality, we can deduce the existence of effective methods to answer both the AST and PAST problem for the corresponding PHORS $\gphors$, by relying on the decidability of the \emph{first order theory of the reals} $\Etor$ (see for example \cite{Marker2002-jm} for an account of this classic result by Tarski). On the other side, we can obtain for suitable classes of PHORS a result in the spirit of the celebrated Chomsky-Schutzenberger enumeration theorem for context-free languages: the generating power series of such PHORS will be shown to be algebraic. This paves the way to studying the asymptotic rate of convergence of such recursion scheme through the techniques of analytic combinatorics, although this topic is not explicitly treated in this paper.

\paragraph*{Algebraic PHORS via Bounded Exponentials}

To obtain such conditions we rely, once more, on the toolbox of linear logic: 
as the relational semantics counts input uses, 
bounding the number of \emph{unknowns} that we need to solve to compute the interpretation of the program amounts at bounding the number of uses that the program may do of each of its inputs.
A first obvious idea would be thus to restrict ourselves to linear, or even affine, programs, as in \cite{DBLP:conf/lics/LiMO22}. Yet, the PHORS in \eqref{eq:phors1} shows that one can well admit {bounded} forms of non-linearity. 
A standard and well-studied way to impose such bounds is via \emph{bounded  exponentials} \cite{GIRARD19921, 10.1007/978-3-642-02273-9_8, 10.1007/978-3-642-54833-8_19, 10.1007/978-3-642-54833-8_18} $!_nA$, where a program $t:\ !_nA\multimap B$ may use its input \emph{at most} $n$ times. 
In our example, the non-terminal $H$ could be given the bounded type $!_2(!_1o\multimap o)\multimap (!_1o\multimap o)$, as $H$ uses its input $f$ twice (indeed, all generating functions $H_i(z)$ are equal to zero, for $i\neq 2$).


But it is possible to go even beyond that, still preserving decidability.
Similarly to how we handled the formal variable $z$ counting probabilistic choices above, 
an algebraic PHORS may well use some inputs in an \emph{unbounded} way, provided these are handled as \emph{formal parameters}.
This idea leverages the intrinsically parametric nature of the weighted relational semantics, which allows one to pass smoothly from power series $s(x,w)$ in $x,w$, with values over some continuous semiring $\Rsemiring$, to ``parametric'' power series $s(w)(x)$ in $x$ with values themselves of the continuous semiring $\fps{\Rsemiring}{w}$ of power series in $x$.
This allows us to admit also infinitely graded types $!_\infty\sigma\multimap \tau$, although used in a restricted way.

\section{PHORS}
In this section we recall Higher Order Recursion Schemes (HORS), their probabilistic counterpart, Probabilistic  Higher Order Recursion Schemes (PHORS), as well as their correspondence with the (probabilistic) $\LY$-calculus. 

\subsection{HORS, a.k.a.~the $\LY$-calculus}
HORS, widely used in higher order model checking, can be seen as grammars that generate infinite ranked trees (or equivalently, infinite typed lambda terms).
In the theory of HORS, the simple types are generated by the grammar $T= \1 \mid T \to T$. We can also define the order of a type $ord(\1) =0$ and $ord(T \to S)= \max(ord(T)+1, ord(S))$.
A (deterministic) HORS $\gphors$ can be defined as a 4-uple $\phors$ where $\mathcal{N}$ is a set of typed non-terminals, $\mathcal{T}$ is a set of typed terminals, $S$ is a distinguished non-terminal called the starting symbol of $\gphors$ and $\mathcal{D}$ is a function associating each non-terminal $L$ with a rewriting rule of the form:
$L x_1 \dots x_n = t$, where $FV(t) \subseteq \{x_1, \dots x_n\} \cup \mathcal{T}$. These rules are to be understood as a (mutually) recursive definition of the non-terminals that can be progressively unfolded: to $\gphors$ we will associate the rewriting system induced by the rule $L t_1 \dots t_n \to t[x_1/t_1, \dots x_n/t_n]$. The rewriting, starting from $S$, will generate at each step finite (simply typed) terms containing the terminals $\mathcal{T}$; the "limit" of this set of terms will be the $\textit{value tree}$ defined by the $\mathcal{G}$ (for formal definitions, see \cite{DBLP:conf/lics/Ong15}, \cite{DBLP:books/ems/21/CarayolS21}). The nodes of this tree will be labelled by terminals and its branching factor at each node will be equal to the arity of the non terminal labelling that node.

The \emph{branch language} $\mathcal L_{\gphors}\subseteq \mathcal T^*$ of an HORS $\gphors$ is the branch language of its value tree: its words are the sequences of terminals encountered in some finite branch of the tree.
 It is well known that the branch languages of order-1 HORS coincide with the context-free languages, while the branch languages of order-2 HORS coincide with the \textit{indexed languages} defined by Aho \cite{DBLP:conf/focs/Aho67}.


%
\begin{example}
	Consider the order-2 HORS $\gphors \bydef \phors$ defined as follows: its non terminals  are $S:\1$, $L: (\1 \to \1) \to \1 \to \1$, its terminals are $e: \1$ or arity 0, $r:\1 \to \1$ of arity 1 and $h:\1 \to \1 \to \1 $ of arity 2. Its rewriting rules are:
	\begin{align*}
		&S = L r e\\
		& L f k = h k (L f (L f k))   
	\end{align*}
	The first reduction steps are $S \to L r e \to h e (L r (L r e))$. Now we can for example unfold the leftmost $L$ and obtain $ h e (h (L r e) (L r (L r (L r e))))$.
\end{example}

The language of HORS is equivalent to the $\LY$-calculus, that is, the simply-typed $\lambda$-calculus enriched with cartesian products and a fixpoint $Y:(T\to T)\to T$: 
in $\gphors=(\nonterm, \C T, \C D, S)$ the function $\C D$ can be equivalently defined as a
function associating each non-terminal $L\in \nonterm$ with a simply typed term 
$\nonterm\vdash \lambda x_1.\dots. x_n. t_L:\nonterm(L)$, such that $t_L$ contains no $\lambda$-abstraction and 
$\nonterm, x_1,\dots, x_n\vdash t_L:\1$. 
To $\gphors$ can we then associate the $\LY$-term
$Yt_{\gphors}:\nonterm$, where 
$t_{\gphors}:= \lambda \langle L_1,\dots, L_n\rangle. \langle  t_{L_1} \dots  t_{L_n}\rangle: \nonterm\to \nonterm$. In the sequel, we will sometimes confuse between the PHORS $\gphors$ and the corresponding term $t_\gphors$, unless this creates ambiguity.

%
%
%

\subsection{PHORS, a.k.a.~the Probabilistic $\LY$-calculus}
PHORS are a probabilistic extension of HORS: they can be obtained by adding to the syntax of HORS a set of new constant symbols $\oplus_p, p \in \mathbb{Q}\cap[0,1]$, of type $\1 \to \1 \to \1$, $t_1 \oplus_p t_2$ representing a probabilistic choice that chooses with probability $p$ the first argument and with probability $1-p$ the others; moreover, as in \cite{DBLP:journals/lmcs/KobayashiLG20}, we focus on PHORS with a unique terminal symbol $e:\1$. A PHORS will thus be given as a triple $\gphors=(\nonterm, \C D, S)$. Each non-terminal $L$ will be given a recursive definitions of the form 
\[L x_1 \dots x_n = t_L \oplus_p t_R,\]
with the two (one step) probabilistic rewriting rules
\begin{align*} 
L t_1 \dots t_n & \redbigp{\mathsf l,p} t_L[x_1/t_1, \dots x_n/t_n],\\
L t_1 \dots t_n &\redbigp{\mathsf r,1-p} t_R[x_1/t_1, \dots x_n/t_n].
\end{align*}
 Here, we annotated the rewriting to keep track of the probability and the direction of this choice. For a reduction $R: s \redbigp{\sigma_1 ,p_1} s_1 \dots  \redbigp{\sigma_k ,p_k} s_k$, we also write $R:s \redbigp{\sigma,p} s_k$, where
 $\sigma = \sigma_1 \dots \sigma_n$ and $p = p_1\dots p_n$.
 
 
 Given a PHORS $\gphors=(\nonterm, \C D,S)$, we can define its probability of termination
 and expected number of steps to termination as
 \begin{align*}
 \Pterm{\gphors}&:= \sum_{ n \: \text{normal form}}\sum_{R:S \redbigp{\sigma, p} n} p \qquad
 \ExpTime{\gphors}&:=  \sum_{ n \: \text{normal form}}\sum_{R:S \redbigp{\sigma, p} n} |\sigma |p
 \end{align*}
We can now define the two central problems of this paper:
\begin{definition}
	Given a PHORS $\gphors$, the AST problem asks to decide whether $\Pterm{\gphors}=1$, while the PAST problem asks to decide  whether $\ExpTime{\gphors} < + \infty$
\end{definition}

\begin{example}\label{ex:phors0}
Consider the order-1 PHORS defined by
\begin{align*}
Fx&= F(Fx)\oplus_{\frac{1}{2}} x,\\
S&= Fe.
\end{align*}
As any choice either adds one $F$ or deletes one, and termination comes when $e$ is reached, it simulates a simple random walk starting from $+1$, adding $\pm 1$ at each step and terminating at $0$. We will see in the next sections that the aforementioned PHORS is AST but not PAST. 
\end{example}

It is known that PAST implies AST. On the other hand, in \cite{DBLP:journals/lmcs/KobayashiLG20}, it is proved that AST is undecidable for PHORS of order at least 2, while a decision procedure exists for PHORS of order 1. In \cite{DBLP:conf/lics/LiMO22}, it is shown that, if we restrict to terms typable via affine types, AST and PAST are decidable. The proof builds upon the correspondence, devised in \cite{DBLP:conf/mfcs/ClairambaultM19}, between affine PHORS and restricted tree stack automata. As a consequence of the theorems of Section 6, we will see that this can be proved without resorting to automata, but rather relying on the relational model of linear logic.

\begin{remark}
	To each PHORS we can associate an HORS whose value tree represents the probabilistic choices made during reduction: to this aim, it is enough to treat every $\oplus_p$ as a terminal symbol of arity 2. This way, we can see that the branching structure of the order-2 PHORS \eqref{eq:phors1} from Section 2, yielding a non-context free branch language, 
	cannot be simulated by any order-1 PHORS.
\end{remark}

The language of PHORS is equivalent to the $\PLY$-calculus, that is, the extension of $\LY$ with (ground) choice operators $M \oplus_p N$, with 
reduction rules
$M\oplus_p N  \redbigp{\mathsf l,p} M$,
$M\oplus_p N \redbigp{\mathsf r,1-p} N$, and typing rule
\[
\infer{\Gamma\vdash M\oplus_p N:\1}{\Gamma\vdash M:\1 & \Gamma\vdash N:\1}.
\]
The \emph{call-by-name reduction $\redbigp{}$} of the $\PLY$-calculus is the closure of the reductions above as well as 
 $\beta$-reduction $(\lambda x.M)N\redbigp{\epsilon,1} M[N/x]$ and $Y$-reduction $YM\redbigp{\epsilon,1}M(YM)$ under the congruence
 $M\redbigp{\sigma, p} N \ \Rightarrow \ MP \redbigp{\sigma, p} NP$.

The translation of a PHORS $\gphors$ into a term $t_{\gphors}$ of the $\PLY$-calculus works as in the case of HORS. In particular, reductions sequences $R:S\redbigp{\sigma, p}n$ to normal form are in one-to-one correspondence with probabilistic call-by-name reductions $R:M_S\redbigp{\sigma, p}n$ to normal form in $\PLY$.
For example, the PHORS from Example \ref{ex:phors0} is encoded in $\PLY$ by $Yt_{\gphors}:(o\to o)\times o$, with $t_{\gphors}=Y \lambda\langle F,S\rangle. t':((\1\to \1)\times \1)\to ((\1\to \1)\times \1)$ and  
$t'=
\langle \lambda  x.F(Fx)\oplus_{\frac{1}{2}} x, Fe\rangle$.

%
%

\section{Interpreting PHORS by Formal Power Series}

In this section, we recall the well-studied relational semantics of linear logic weighted over a continuous semiring $\Rsemiring$, in a slightly different flavor than, say, \cite{DBLP:conf/lics/LairdMMP13}: rather than using $\Rsemiring$-valued matrices, we will use formal power series over $\Rsemiring$, as in \cite{DBLP:journals/corr/abs-2501-15637}. This kind of idea had already been outlined in a more categorical setting in \cite{LAMARCHE199237}. We will then illustrate how this semantics associates each PHORS with a generating function capturing its probability of termination.



\subsection{Formal Power Series}
Let us first introduce formal power series (fps in short); for more details about them and their applications to combinatorics we refer to \cite{DBLP:books/daglib/0023751}, \cite{DBLP:series/tmsc/KauersP11}. Given a set $\Sigma$, 
let $!\Sigma$ be the set of finite multisets over it, i.e functions $\mu: \Sigma \to \N$ with finite support. Given a semiring $\Rsemiring$, the set of \emph{formal power series (with commuting variables) over $\Rsemiring$}, denoted $\fps{\Rsemiring}{\Sigma}$, is the set of all functions $!\Sigma \to \Rsemiring$. More concretely, if
 we introduce for each $s \in \Sigma$ a variable $x_s$ (we will denote by $x_\Sigma$ the set of all these variables), then a finite multiset $\mu \in !\Sigma$ can be seen as the monomial $x_\Sigma^\mu \bydef \Pi_{s \in \Sigma}x_s^{\mu(s)}$; then any formal power series $s\in \fps{\Rsemiring}{\Sigma}$ can then be expressed as a formal sum:
 $$s=\sum_{\mu \in !\Sigma} s_\mu x_\Sigma^{\mu}.$$ 
 We will sometimes write $s(x_\Sigma)$ to underline which variables appear in $s$.
 For each $s\in \fps{\Rsemiring}{\Sigma}$, its \emph{support} $\supp~s\subseteq \ !\Sigma$ is the set of multisets $\mu$ such that $s_{\mu}\neq 0$. 
 We denote as $\fpp{\Rsemiring}{\Sigma}\subseteq \fps{\Rsemiring}{\Sigma}$ the set of \emph{polynomials}, i.e.~of fps with finite support, which can written, as usual, as finite sums

%

  \begin{example}
  The power series $s(x)=\sum_{n} (1/2)^n x^n$ belongs to $\fps{\mathbb Q}{x}$.
  The power series $s(x)=\sum_{n=0 }^{\infty}\sum_{i+j=n} (1/3)^{n} x_0^{i}x_1^j$, i.e. \\
$ \sum_{\mu\in !\{0,1\}}(1/3)^{\mu(0)+\mu(1)}x^{\mu}$, belongs to $\fps{\mathbb Q}{x_0,x_1}$.
  \end{example}
 
 In $\fps{\Rsemiring}{\Sigma}$ we can define two operations: sum, performed componentwise: $\sum_{\mu \in !\Sigma} r_\mu x_\Sigma^{\mu} + \sum_{\mu \in !\Sigma} s_\mu x_\Sigma^{\mu} \bydef \sum_{\mu \in !\Sigma} (r_\mu + s_\mu)  x_\Sigma^{\mu}$ and the Cauchy product: 
 $
 \sum_{\mu \in !\Sigma} r_\mu x_\Sigma^{\mu} \cdot \sum_{\mu \in !\Sigma} s_\mu x_\Sigma^{\mu} \bydef \sum_{\kappa \in !\Sigma}\left (\sum_{\mu + \nu = \kappa} r_\mu s_\nu\right)  x_\Sigma^{\kappa}$.
 With these operations, $\fps{\Rsemiring}{\Sigma}$ is a semiring; if we take on it the pointwise partial order, it becomes a \emph{continuous} semiring: this means that all directed joins have a supremum, and such suprema commute with multiplication.
Our typical example here is the continuous semiring of positive extended reals $\Rinf$.   Continuity is essential when we want to define the \textit{composition} of formal power series. Take a power series $r \in \fps{\Rsemiring}{\Sigma}$ and let $s_\Sigma\in \fps{\Rsemiring}{\Sigma'}^{\Sigma}$ be a $\Sigma$-indexed family of power series over the set $\Sigma'$, $s_\sigma = \sum_{\nu \in !\Sigma'}s_{\sigma, \nu} y^\nu$. Then, we can define the power series $r(s_\Sigma) \in \fps{\Rsemiring}{\Sigma'}$ by the formula:

\vskip-3mm
{\small
   \begin{align*}
   r(s_\Sigma)=
\sum_{\kappa \in ! \Sigma'}  \left( \sum_{\mu=[\sigma_1, \dots, \sigma_k] \in !\Sigma} \sum_{
{\tiny
\begin{matrix}
(\nu_1,\dots,\nu_{k})\in (!\Sigma')^k\\
\nu_1+\dots+\nu_{k}=\kappa
\end{matrix}
}
}
r_{\mu }\cdot 
 \prod_{i=1}^{k}  s_{\sigma_i, \nu_i } \right)y_\Sigma^\kappa
  \end{align*}
  }
  Observe that if there is at least an $s_\sigma$ such that $s_{\sigma, []} \neq 0$, the sum over $!\Sigma$ will be infinite; still, by continuity, we can define its value to be the sup of the partial sums (yet, notice that, over a - non-continuous - ring, this composition would not be well-defined). If the power series $s_\Sigma$ are all constants (i.e. $s_{\sigma, \mu}=0 \: \forall \mu \neq []$), we will say that $r(s_\Sigma) \in \fps{\Rsemiring}{\emptyset}= \Rsemiring$ is the \emph{value of $r$ at the point} $(s_{\sigma, []})_{\sigma \in \Sigma}$.
With respect to these operations, $\fpp{\Rsemiring}{\Sigma}$ form a (non-continuous) sub-semiring of $\fps{\Rsemiring}{\Sigma}$.

Since $\Qsemiring:=\fps{\Rsemiring}{\Sigma}$ is a continuous semiring, for any set $\Sigma'$ we can consider the continuous semiring $\fps{\Qsemiring}{\Sigma'}$ of formal power series whose coefficients are themselves power series (on $\Rsemiring$), and we have the isomorphism $\fps{\Qsemiring}{\Sigma'}=\fps{(\fps{\Rsemiring}{\Sigma})}{\Sigma'}\equiv\fps{\Rsemiring}{\Sigma+\Sigma'}$: this generalizes the well-known remark that a fps in two variables $s(x,y)$ can always be written as a fps $s_y(x)=\sum_n s_n(y)x^n$ in $x$ whose coefficients are fps $s_n(y)$ in $y$.

If we do not start, as we do here, with a continuous semiring $\Rsemiring$, but rather with a ring $R$, almost all constructions still work: we can similarly define the set $\fps{R}{\Sigma}$ of formal power series over $R$, which is a ring, as well as its subring  $\fpp{R}{\Sigma}$ of polynomials over $R$. Still, the composition of formal power series will not be defined in general, as it involves an infinite sum, while the composition of polynomials remains well-defined.

  \subsection{The Weighted Relational Semantics}
  
    Now, we recall the link between formal power series and the weighted relational model of linear logic \cite{DBLP:conf/lics/LairdMMP13}. Given a continous semiring $\Rsemiring$, the category 
     $\Qrelkleisli{\Rsemiring}$ has sets as objects, and morphisms $\Qrelkleisli{\Rsemiring}(X, Y)$ are $\Rsemiring$-valued matrices indexed by $!X \times Y$. 
     $\Qrelkleisli{\Rsemiring}$ is the coKleisli category, with respect to the $!$ comonad, of the more familiar category  $\Qrel{\Rsemiring}$ of sets and $\Rsemiring$-valued matrices. 
    
     At the same time, $\Qrelkleisli{\Rsemiring}$ can be seen as a category of formal power series:    a matrix $(t_{\mu, y})_{\mu \in !X, y \in Y}$ can be identified with the $Y$-indexed family of fps $(\sum t_{\mu, y} x_X^\mu)_{y \in Y}$, the morphisms of $\Qrelkleisli{\Rsemiring}$ can be identified with (families of) formal power series, i.e.~$\Qrelkleisli{\Rsemiring}(X,Y)\equiv\fps{\Rsemiring}{X}^Y$.
    Indeed, composition in $\Qrelkleisli{\Rsemiring}$ is given by (pointwise) composition of the underlying power series, with identities $\mathrm{id}_X\in  \fps{\Rsemiring}{X}^X$ given by
    $\mathrm{id}_i(x_X)=x_i$.
%

$\Qrelkleisli{\Rsemiring}$ is cartesian closed, with cartesian products given by $X+Y$ (with neutral $0:=\emptyset$) and exponentials given by $!X\times Y$. 
This allows us to interpret simple types as $\model o=1:=\{\star\}$, $\model{T\times U}=\model T+\model U$ and $\model{T\to U}=!\model{T}\times \model{U}$. Intuitively, a type $T$ translates into a set of variables, and a term $\Gamma\vdash t:T$ into a $\model{T}$-indexed family of fps $\model{t}^{\Rsemiring}_i(x_\Gamma)$ in \emph{as many variables as $\model{\Gamma}$}.
Here we can observe the main challenge appearing with higher-order types, as these are interpreted by \emph{infinitely many} variables: for instance, $\model{o\to o}=!\model o\times \model o\equiv\N\times 1\equiv \N$ translates into a countable sequence of variables, and 
$\model{(o\to o)\to (o\to o)}\equiv !\N\times\N$ has one distinct variable $x_{\mu,n}$ for each $\mu\in !\N$ and $n\in \N$.
Higher-order terms correspond thus to fps in \emph{countably many} variables.

%

\begin{example}\label{ex:churchtwo}
Consider the program $t=\lambda x.y(yx)$, where $y:(\1\to\1)\vdash t:\1\to\1$: recalling $\model{\1\to\1}=!\1\times \1\equiv \N$, 
we have that $t$ is interpreted by a $\N$-indexed family of power series
$(b_i(y_{\N}))_{i\in\N}\in \Qrelkleisli{\Rsemiring}(\N,\N)=\fps{\Rsemiring}{\N}^{\N}$.
Now, think of the variables $y_{\N}$, which interpret the function $y:\1\to\1$, as encoding the fps $a(x)=\sum_n y_n x^n\in\Qrelkleisli{\fps{\Rsemiring}{\N}}(1,1)\equiv  \fps{(\fps{\Rsemiring}{\N})}{1}$;
the term $t$ should translate then into the composition 

{\small
\[
a(a(x))=\sum_n y_n\left(\sum_m y_mx^m\right)^n=
\sum_{i=0}^{\infty}
\left(
\sum_{{\tiny\begin{matrix}(n,m_1,\dots, m_n),\\ m_1+\dots+m_n=i\end{matrix}}}y_ny_{m_1}\dots y_{m_n}\right) x^i,
\] 
}
which gives $b_i(y_{\N})= \sum_{{\tiny\begin{matrix}(n,m_1,\dots, m_n),\\ m_1+\dots+m_n=i\end{matrix}}}y_ny_{m_1}\dots y_{m_n}$.
Notice that the evaluation $tx$ of $t$ over some variable $x:o$ is then precisely interpreted by $a(a(x))\in \Qrelkleisli{\Rsemiring}(\N+1,1)=\fps{\Rsemiring}{\N+1}\equiv\fps{(\fps{\Rsemiring}{\N})}{1}$.

%
%
%
\end{example}

Beyond exponentials, 
$\Qrelkleisli{\Rsemiring}$ is endowed with Conway fixpoints \cite{DBLP:journals/mscs/Hasegawa09}, \cite{DBLP:conf/mfcs/GrelloisM15}: given a morphism $f \in \Qrelkleisli{\Rsemiring}(X + Y, X)$,  we define $\fix f \in \Qrelkleisli{\Rsemiring}(Y, X)$ as follows: take the sequence $f^0 \bydef  0 \times id_Y \in \Qrelkleisli{\Rsemiring}(1 \times Y ,X \times Y), \: f^{n+1} \bydef f \circ (f^n \times id_Y) \circ \langle id_\1, \Delta_Y \rangle  \in \Qrelkleisli{\Rsemiring}(1 \times Y, X)$ and finally let $\fix_{Y,X} f \bydef \sup_n f^n \in \Qrelkleisli{\Rsemiring}( Y, X)$.
   It is worth to restate this construction in terms of power series: $f$ will be represented by an $X$-indexed family $(s_x(x_X, x_Y))_{x \in X}$. Then we can define its iterates and the fixpoint as follows, for $x \in X$:
   \begin{equation}\label{eq:fixpointeq}
   \begin{aligned}
   	& r_x^{(0)}(x_Y) = 0 \in \fps{\Rsemiring}{Y}
   	& r^{(n+1)}_x(x_Y) =  s_x(r_X(x_Y), x_Y)\\
   	& (\fix s_X)_x(x_Y) = \sup_n r^{(n)}_x(x_Y)
   \end{aligned}
   \end{equation} 
	From this, it is clear that the power series $r_X= \fix f$ are the minimal solution of the infinite family of equations: $(r_x = s_x(r_x, x_Y))_{x \in X}$.

Using the cartesian closed structure and the fixpoints, any term 
$\Gamma\vdash t:T$
 of $\LY$ yields in $\Qrelkleisli{\Rsemiring}$
a family of fps $\model{t}^{\Rsemiring}\in \fps{\Rsemiring}{\model{\Gamma}}^{\model{T}}$.
To interpret $\PLY$, we restrict our attention to semirings of the form $\Rsemiring=\fps{\Rinf}{\{z\}+\Sigma}$, where $z$ denotes a distinguished variable. This allows us to interpret probabilistic choice (with bias $p$) as 
\[
\model{M\oplus_p N}^{\Rsemiring}=p z\cdot\model{M}^{\Rsemiring}+(1-p) z\cdot\model{N}^{\Rsemiring}.
\]
As shown by Proposition \ref{prop:proba} below, the variable $z$ plays the role of a \emph{counter} for each probabilistic choice: each reduction $t\redbigp{\sigma, p}e$ will produce a monomial $pz^{|\sigma|}$ in the semantics.
%
%
%


\begin{remark}\label{rem:tropical}
\cite{DBLP:journals/corr/abs-2501-15637} considers the interpretation of
\emph{parametric} choices  $M\oplus_x N$, i.e.~choices according to some unknown bias $x$, by taking the semiring 
  $\Qsemiring=\fps{\Rsemiring}{x,\overline x}$ and letting $\model{M\oplus_p N}^{\Qsemiring}=x\cdot M+\overline{x}\cdot N$. \end{remark}

%


We now see how any PHORS $\gphors=(\nonterm,\C D,S)$ produces a 
morphism $\model{\gphors}^{\fps{\Rinf}{z}}$ (i.e.~$\model{t_{\gphors}}^{\fps{\Rinf}{z}}$) in
$\Qrelkleisli{\fps{\Rinf}{z}}(\model\nonterm,\model\nonterm)\equiv \fps{(\fps{\Rinf}{z})}{\model\nonterm}^{\model\nonterm}$, i.e.~a $\model{\nonterm}$-indexed family of power series in the variables $x_{\model{\nonterm}}$.
All this leads to the following definitions:
\begin{definition}[PGF of a PHORS]
For every PHORS $\gphors=(\nonterm, \C R,S)$ and non-terminal $L_i:T_1\to\dots\to T_n\to o\in\nonterm$, letting $\Sigma=\model{T_1}+\dots+\model{T_n}$, we define:
\[
a_{L_i}(z)(x_{\Sigma}):=\pi_i\left(
\mathsf{fix}_{\nonterm}\model{{\gphors}}^{\fps{\Rinf}{z}}\right)\in \fps{(\fps{\Rinf}{z})}{\Sigma}.
\]  
In particular, the fps $a_{\gphors}(z):=a_S(z)\in \fps{\Rinf}{z}$ is called
the
 \emph{probabilistic generating function of $\gphors$}.

\end{definition}
Observe that the fps $a_{L_i}(z)(x_\Sigma)$ are the minimal solutions of the equational system $(r_x=(\model{\gphors}^{\Rsemiring})_x(r_x,z))_{x\in \model\nonterm}$ induced by $\gphors$. 

The generating function $a_\gphors(z)$ precisely captures the call-by-name probabilistic execution of closed terms, as stated below:
\begin{proposition}\label{prop:proba}
For any PHORS $\gphors$,
$a_{\gphors}(z)=\sum_{i=0}^{\infty}\mathbb P(\gphors\downarrow_i)z^i$
where $\mathbb P(\gphors\downarrow_i)=\sum_{R:t\to^i e}w(R)$ is the probability that $S$ terminates after exactly $i$ probabilistic steps. In particular, $a_\gphors(1)=\mathbb P(\gphors\downarrow)$. 
\end{proposition}
\begin{proof}
From \cite{DBLP:conf/lics/LairdMMP13, DBLP:journals/jacm/EhrhardPT18,DBLP:journals/corr/abs-2501-15637} we know that each reduction $S\redbigp{\sigma,p}e$ precisely adds up one monomial $pz^{|\sigma|}$, hence $a_{|\sigma|}z^{|\sigma|}$ in $a_\gphors(z)$ is the sum of the probabilities of all reductions of length $|\sigma|$.
\end{proof}
From Proposition \ref{prop:proba} we also deduce that the \emph{derivative} of $a_\gphors(z)$ captures the expected number of steps to termination:

\begin{corollary}\label{cor:expected}
$a'_\gphors(1)=\sum_{i=1}^{\infty} i\cdot \mathbb P(\gphors\downarrow_i)=\mathbb E(\gphors\downarrow)$.

\end{corollary}
%


%

\begin{example}\label{ex:phors2}
Consider the PHORS from Example \ref{ex:phors0}, with $t_F=F(Fx)\oplus_{\frac{1}{2}}x$ and $t_S=Fe$.
%
The interpretation in $\Qrelkleisli{\fps{\Rinf}{z}}$ 
 of $\lambda \langle F,S\rangle. t_F:((\1\to \1)\times 1)\to (\1\to \1)$ is a $\N$-indexed family of fps 
$s^F_i(y_{\N+1})\in \fps{(\fps{\Rinf}{z})}{\N +1}$, given, for $i\in \N$, and recalling $b_i(y_{\N})$ from Example \ref{ex:churchtwo}, by
\[
s^F_i(y_{\N+1})=
\begin{cases}
\frac{1}{2}zb_1(y_{\N})+\frac{1}{2}z=\frac{1}{2}(z
y_1^2+z)
& \text{ if }i=1,\\
\frac{1}{2}zb_i(y_{\N})
&\text{ if }i\in\N, i\neq 1,
\end{cases}
\]
The interpretation of $\lambda \langle F,S\rangle. t_S:((\1\to \1)\times 1)\to \1$ is a unique fps 
$s^S(y_{\N +1})\in \fps{\Rinf}{\N +1}$ given by $s^S(y_{\N +1})=\sum_{i\in \N} y_i$. 
 
Computing $ \mathsf{fix}\  \langle s^F,s^S \rangle$ means finding a minimal solution $(a_i(z))_{i\in \N+1}\in(\fps{\Rinf}{z})^{\N+1}$ of the fixpoint equations $a_{i}(z)=s^F_i(a_{\N+1}(z))$ ($i\in\N$) and $a_{\gphors}(z)=s^S(a_{\N +1}(z))$.
One can easily see that this yields $a_{i}(z)=0$, for $i\neq 1,\star$, while $a_\gphors(z):=a_\star(z)=a_1(z)\in \Rinf$ can be found as minimal solution of the polynomial equation
\begin{equation}
\label{eq:alg1}
a_\gphors(z)=\frac{1}{2} \left (za_\gphors^2(z)+z\right)
\end{equation}
%
%
We will see in the next section that this solution is given by the series
$a_\gphors(z)=\sum_{i=0}^{\infty}\frac{C_i}{2^{2i+1}} z^{2i+1}$, where $C_i$ is the $i$-th Catalan number.
\end{example}

 For order-1 PHORS, we have the following useful lemma:
\begin{lemma}\label{lemma:affine}
For all order-1 PHORS $\gphors=(\nonterm, \C R,S)$ and non-terminal symbol $L\in \nonterm$, the fps $a_L\in \fps{(\fps{\Rsemiring}{z})}{x_1,\dots, x_n}$ is \emph{affine}:
%
 there exist $w_0,w_1,\dots, w_n\in\fps{\Rinf}{z}$ such that
\[
a_{L}(z)(x_1,\dots, x_n)
=w_0(z)+w_1(z)x_1+\dots+w_n(z)x_n,
\]
where $w_0(1)=\mathbb P[L\vec x\to^* e]$ and
$w_{i+1}(1)=\mathbb P[L\vec x\to^* x_{i+1}]$.
\end{lemma}
This result translates the fact that in a reduction of $Lx_1\dots x_n$ to normal form, a ground variable $x_i:\1$ can occur in head position \emph{at most once}: as soon as it does, we must have $Lx_1\dots x_n\to^* x$, that is, the reduction has terminated.

%

\section{Algebraic Power Series}

  As we saw, power series are infinitary objects and their manipulation involves the concept of limit, making them generally uncomputable. In this section, we explore a way to provide  \emph{finitary} specifications for power series through the concept of algebraicity. 
  Algebraic power series are widely used in algebra and combinatorics, in the context of rings or fields. Their treatment in semirings is less standard, but extensively studied in the context of formal language theory.
    We will show that, through the weighted relational semantics, such ideas apply naturally to  probabilistic $\lambda$-calculi.

\subsection{Fixpoint Algebraic Systems}
  
In the previous section we have seen that PHORS can be associated with formal power series defined as the minimal solution of a (generally infinite) equational system. In practice, computing such power series directly can be very hard; recall, by the way, that such fps capture properties like AST or PAST, which are undecidable in general. However, when the equational system defining a family of power series can be expressed by \emph{finitely many polynomial equations}, it becomes possible to extract information about such power series
as well as the sequences of their coefficients.

As a famous example, consider again the fps $c(x)= \sum_{n \geq 0} C_n x^{n} \in \fps{\R}{x}$. $c(x)$ is the Taylor expansion around $0$ of the (analytic) function $\frac{1-\sqrt{1-4x}}{2x}$. From this, it is easy to see, taking $p(w, x)= xw^2 - w + 1$, that $p(c(x), x)=0$. Observe that, while computing e.g.~$c(1/4) = \sum_{n \geq 0} \frac{1}{4^n(n+1)} \binom{2n}{n}$ as a limit is hard, we can easily deduce that $c(1/4)=1$ is the (only) positive root of $p(w,1/4)$.

\begin{example}\label{ex:phors4}
Equation \eqref{eq:alg1} from Example \ref{ex:phors2} can be rewritten as $p(a_\gphors(z),z)=0$, where $p(w,z)=\frac{1}{2}zw^2-w+\frac{1}{2}z$, which gives the solution $a_\gphors(z)=\frac{1-\sqrt{1-z^2}}{z}$.
Using the fps $c(x)$ from above, we deduce, with $x=\left(\frac{z}{2}\right)^2$, 
$a_\gphors(z)=\frac{1-\sqrt{1-4z}}{2z}z=c(z)\frac{z}{2}=\sum_{i=0}^{\infty}\frac{C_i}{2^{2i+1}}z^{2i+1}$.
The interpretation of this is that, for any $i\in \N$, there are $C_i$ terminating reductions that make $2i+1$ probabilistic choices. On the other hand, there is no terminating path of even length, as there is no random walk going from 1 to 0 in an even number of steps.
Now, as $1$ is the smallest root of $p(w,1)$, we easily obtain $\mathbb P[\gphors\downarrow]=a_{\gphors}(1)=1$, that is, that $\gphors$ is AST; moreover, $\mathbb E[S\downarrow]$ is given by the diverging series $a'_{\gphors}(1)=\sum_{i=0}^{\infty}\frac{2i+1}{2^{2i+1}}C_i=
\infty$, that is, PAST fails.
\end{example}

  Before considering algebraic power series on a (continuous) semiring, let us recall the case of rings, indeed the one most customary in combinatorics. We assume our rings to be integral domains.
  \begin{definition}\label{def:fas1}
  	Let $R$ be a ring  and let $R'\leq R$ a subring.
  	A tuple $(s_1, \dots s_n)$ of formal power series $s_i \in \fps{R}{\Sigma}$ is an $R'$-\emph{algebraic family} if there exist polynomials $p_i \in \fpp{R'}{w_1, \dots w_n, \Sigma}$ such that:
  	\begin{equation*}
  	\begin{cases}
  		p_1(s_1, \dots s_n, x_\Sigma)=0\\
  		\dots\\
  		p_n(s_1, \dots s_n, x_\Sigma)=0\\
  	\end{cases}
  	\end{equation*}
  	 When $n=1$, an $R'$-algebraic family $(s)$ is simply called a 
	 \emph{$R'$-algebraic fps} and the polynomial $p_1$ is called an \emph{annihilating polynomial} of $s$.
%
  \end{definition}

  \begin{example}
		Let $s \in \fps{\R}{x}$ be $\sum_{n \geq 0}x^n$. Taking the polynomial $p(w,x) = w(1-x)-1$, we get $p(s(x), x)=0$. Hence, $s$ is algebraic over $\Q$. Indeed, $s$ is the multiplicative inverse $(1-x)^{-1}$ of $(1-x)$ in $\fps{\R}{x}$. In general, every $\Q$-rational function (i.e every power series of the form $r_1(x_\Sigma)r_2(x_\Sigma)^{-1}$ for $r_1, r_2 \in \fps{\Q}{\Sigma}$) is algebraic over $\Q$. 
  \end{example}	
  \begin{remark}\label{rem:hadamard}
  The non-zero coefficients of an algebraic power series cannot be ``too far'': if $a(x)=\sum_na_n x^{p_n}$ is algebraic, then $\lim_n p_n/n<\infty$: this is a consequence of \emph{Fabry's gap theorem}, which asserts that, when $\lim_n p_n/n$ diverges, $a$ cannot be analytically continued on any point of its circle of convergence.
%
%
  For instance, no series of the form $\sum_n a_n x^{n^2},\sum_n a_n x^{n^3}, \sum_n a_nx^{2^n}$ is algebraic.
  \end{remark}

  \begin{remark}
  	\label{derivative-rational}
  	If $s\in \fps{\R}{x}$ is algebraic with an annihilating polynomial $p(w, x_\Sigma)$,
also its derivative $s'\in \fps{\R}{x}$ is algebraic. Moreover, one can find effectively (\cite{CHUDNOVSKY1986271}, \cite{DBLP:series/tmsc/KauersP11}) a rational function $r(x,y)$ such that $s'(x)$ can be computed from $s$ via 
 $s'(x) = r(x, s(x)) $.
  \end{remark}
    \begin{remark}
  	\label{coefficients-rational}
	The coefficients $a_n$ of a unary algebraic fps
$a(x)=\sum_n a_n x^n$ whose annihilating polynomial $p(w, x)$ satisfies $p(0,0)= 0$, $\partial_wp(0,0)\neq 0$ and $p(0,z)\neq 0$ can be computed directly via a ``multinomial formula'' \cite{BANDERIER_DRMOTA_2015}.
More generally, since all algebraic functions are \emph{D-finite} \cite{DBLP:series/tmsc/KauersP11}, one can deduce a linear recursive equation with polynomial coefficients 
$p_0(n)a_n+p_1(n)a_{n+1}+\dots +p_r(n)a_{n+r}=0$ that can be used to
compute the $a_n$ effectively.
   \end{remark}

  The theory of algebraic power series should be stated in a slightly different way in the context of semirings: if $a \in \fps{\Rsemiring}{\Sigma}, a > 0$, then for each $k$ $a^k > 0$, hence it is not possible that $p(a, x_\Sigma)=0$ for some $p \in \fpp{\Rsemiring}{\Sigma}$. Definition \ref{def:fas1} must thus be adapted as follows:
     \begin{definition}[\cite{DBLP:books/sp/KuichS86}]
     	\label{def:fas2}
	Let $\Rsemiring$ be a semiring and let $\Rsemiring'\leq \Rsemiring$ be a subsemiring. A  $\Rsemiring'$-\emph{fixpoint algebraic system} (in short $\Rsemiring'$-\emph{FAS}) with parameters $x_\Sigma$ is a system of $n$ equations:
	 \begin{equation}
		\label{fixpointsystem}
		\begin{cases}
			w_1 = p_1(w_1, \dots w_n, x_\Sigma)\\
			\dots\\
			w_n = p_n(s_1, \dots w_n, x_\Sigma)\\
		\end{cases}
		\end{equation}
	A family $(s_1, \dots s_n)$ of formal power series $s_i \in \fps{\Rsemiring}{\Sigma}$ that is the minimal solution of a $\Rsemiring'$-FAS is called a $\Rsemiring'$-\emph{fixpoint algebraic family}. 
	If the $p_i$ are such that $p_i(0,0)=0$ and, for each $j$, the coefficient of the monomial $w_j$ in $p_i$ is $0$, we say that the system is \emph{proper}.
\end{definition}
	The theory of systems like \eqref{fixpointsystem} has been extensively studied (\cite{DBLP:books/sp/KuichS86}, \cite{DBLP:journals/cpc/BanderierD15}) in the context of combinatorics and formal language theory: in particular, it can be shown that they admit a (unique) minimal solution, that can be obtained by iterating the polynomials $p_1 \dots p_n$. In the following we will be in general be interested in ${\Rinf}$-power series algebraic over $\Q_{\geq 0}$: in this case, the polynomial equations $w_i= p_i, \: p \in \fpp{\Qpos}{\Sigma}$ can be rewritten as $p'_i \bydef p_i- w_i = 0, \: p \in \fpp{\Qpos}{\Sigma}$. Rather than speaking of the minimal solution of \eqref{fixpointsystem} as a system with coefficients in $\Qpos$, we can then speak of the minimal non-negative solution $(s_1, \dots s_n)$  of the system $(p'_i=0)_{1 \leq i \leq n}$ with coefficients over $\mathbb{Q}$.
	In this case, if the FAS is proper, even more is true: every $s_i$ part of a $\Q^+$-algebraic family is algebraic:
 	\begin{theorem}[Variable Elimination, \cite{DBLP:books/sp/KuichS86}, Thms.~16.9, 16.10]
 		\label{th:elimination}
 		Let $F$ be a proper $\Q_{\geq 0}$-FAS consisting of $n$ equations with parameters $x_\Sigma$ and let $(s_1, \dots s_n)$ be its minimal solution. From $F$ we can  compute irreducible polynomials $(q_i(w' , x_\Sigma))_{1 \leq i \leq n}$ such that:
 		$q_i(s_i(x_\Sigma), x_\Sigma)=0$, for all $1 \leq i \leq n$, 
 		that is, every $s_i$ is a $\Q$-algebraic power series.
 	\end{theorem}
 	\begin{remark}
 		The definition of fixpoint algebraic family consists of two conditions: the purely algebraic fact of being solution of a FAS and the minimality with respect to the order relation. For \emph{proper} FAS, the minimal solution $(s_1(z), \dots s_n(z)) $ is also the unique solution in a neighbourood of $z=0$ \cite{DBLP:books/sp/KuichS86}. Unfortunately, the algebraic functions $s_1(z), \dots s_n(z)$ might have singularities over the disk $|z| \leq 1$, hence there might be multiple solutions of the FAS for $z=1$; to select amongst them the one corresponding to the fixpoint solution, we need to use the minimality. As an example, take the equation $w= zw$: the only solution on the open disk $|z|<1$ is the identically zero power series $w(z)= 0$, but for $z=1$ we find that any $w \in \R$ is a valid solution.
 	\end{remark}
	Over any semiring $\Rsemiring$, if $a(x),b(x)\in \fps{\Rsemiring}{x_\Sigma}$ are algebraic, then so are their sum $a(x)+b(x)$, multiplication $a(x)b(x)$, and derivative $a'(x)$ \cite{DBLP:series/tmsc/KauersP11}. If $\Rsemiring$ is continuous, the same holds of their composition $a(b(x))$ (while, as we saw in Section 4, composition is not always well-defined over arbitrary (semi)rings). 

 	\begin{remark}
 	\label{remark:semilinear}

%
		When $a(x) =  \sum_{\mu \in !\Sigma} a_\mu x^\mu \in \fps{\Rinf}{x_\Sigma}$ is part of a $\Q^+$-proper algebraic family, we can refine what we said in Remark \ref{rem:hadamard}: as a consequence to (\cite{DBLP:books/sp/KuichS86}, Thm.~16.35), the support 
 		$\supp a=\{\mu \mid a_\mu \neq 0\} $
 		is \emph{semilinear}, that is, it is
%
		the union of finitely many sets of the form $\{n_1 \mu_1 + \dots + n_k \mu _k  \mid n_1, \dots n_k \in \N \}$.
 	\end{remark}

\subsection{From PHORS to FAS}

	We are interested in the case in which the fixpoint equations that define (the interpretation of) a PHORS in $\Qrelkleisli{\Rsemiring}$ form a FAS, thus yielding, as solution, an algebraic power series. Recall that higher-order types $T\to U$ are interpreted by \emph{infinitely many} variables and, consequently, higher-order terms translate into infinitely many equations.
Our goal is to find ways to {reduce} such infinitary systems to more finitary (and manageable) ones.

	   As we saw, a PHORS $\gphors$ is interpreted via the fps $a_{L}(z)(x_\Sigma)$ interpreting each non-terminal $L\in \nonterm$; these fps are collectively defined as the minimal solutions of the family of equations
$(r_x = (t_\gphors)_x(r_x, z))_{x \in X}$
induced by the morphism $\fps{(\fps{\Rinf}{z})}{\model{\nonterm}}^{\model{\nonterm}}$ interpreting $t_\gphors:\nonterm\to \nonterm$.

When $\gphors$ is order-1, this system can \emph{always} be reduced to a FAS:
\begin{proposition}[order-1 PHORS are algebraic]
For all order-1 PHORS $\gphors=(\nonterm, \C R,S)$ and non-terminal symbol $L\in \nonterm$, the fps $a_{L}(z)(x_1,\dots, x_n)$ is $\mathbb Q$-algebraic.
\end{proposition}
\begin{proof}
When $\gphors$ is order-1, each type $\nonterm(L_i)$ is of the form $o\to\dots \to o\to o$:
 hence $a_L(z)(x_1,\dots, x_n)$ can be written in the form $\sum_{(p_1,\dots,p_n)\in \N^n} w_{p_1,\dots, p_n}(z)x_1^{p_1}\dots x_n^{p_n}$. 
Thanks to Lemma \ref{lemma:affine} we have that, if for some $i=1,\dots, n$, $p_i\geq 2$, then $ 
w_{p_1,\dots, p_n}(z)=0$. We can thus reduce ourselves to a finite polynomial system only involving the finitely many $w_{p_1,\dots, p_n}(z)$, with the $p_i\in \{0,1\}$.
\end{proof}
Since, as we will see in \ref{th:decidability}, AST and PAST can be decided once a PHORS is interpreted as a FAS, we see that the well-known decidability of order-1 PHORS naturally follows from the simple nature of their generating functions, which are always affine.  

Beyond order-1, this argument does not work anymore, as PHORS can be truly non-linear.
We thus need to restrict ourselves to systems satisfying a suitable form of finiteness, defined below:

\begin{definition}[variable and family restrictions]
For all fps $s\in \fps{\Rsemiring}{\Sigma}$ and $\Sigma'\subset \Sigma$,
 $s\vert_{\Sigma'}\in \fps{\Rsemiring}{\Sigma}$ is the fps $t$ defined by 
$t(x_{\Sigma'},x_{\Sigma-\Sigma'})=s_\sigma(x_{\Sigma'},0)$.
%

%
 \end{definition}
 Intuitively, in $s\vert_{\Sigma'}$ we only keep those monomials $s_\mu x^\mu$ from $s$ such that $\mu$ has support included in $\Sigma'$, and we 
 ``kill'' all monomials $s_\mu x^{\mu}$ containing 
 some $x^i$, with $i>0$ and $x\notin \Sigma'$. 

	\begin{definition}[finitary fps]\label{def:finitary}
		Let $s_X\in \fps{\Rsemiring}{X+Y}^Z$ be a family of fps. A pair $(U,V)$, where $U\subseteq X$ and $W\subseteq Z$, is called:
	\begin{varitemize}
	\item \emph{stable} if for all $\sigma\in Z-W$, $s_\sigma \vert_{U+Y}=0$;
	\item \emph{polynomial} if for all $\sigma\in W$, $s_\sigma \vert_{U+Y}$ has finite support (i.e.~it is a polynomial).
	\end{varitemize}
A family $s_X\in \fps{\Rsemiring}{X+Y}^X$ is called \emph{finitary} if there exists some finite $W\subset_{\mathrm{fin}}X$ such that $(W,W)$ is stable and polynomial.	
If, moreover, for each $\sigma \in W$ and for each $y \in Y$, $s_\sigma(0)=0$ and the coefficient of $y$ in $s_\sigma$ is $0$, we say that $s_X$ is a \emph{finitary proper family}. 	
	\end{definition}

%
%

Take a family $s_X\in \fps{\Rsemiring}{X+Y}^X$ and let $A\subseteq X$ be such that $(A,A)$ is a stable pair. 
By stability, starting from $0$ and repeatedly applying $s_X$ we can never obtain a term of the form $s_\sigma(t)$, with $\sigma\notin A$. Moreover, if $A$ is also polynomial, then all such iterates are obtained by composing polynomials $s_\sigma\vert_{A+Y}$, with $\sigma\in A$; finally, if $A$ is finite, these polynomials are only finitely many. All this leads to the following:
%
%
%
	\begin{proposition}\label{prop:fintoalg}
		 Let $\Rsemiring'\leq\Rsemiring$.  
		 If $s_X\in \fps{\Rsemiring'}{X+Y}^X$ is a finitary family, then $\fix s_X$ has finite support and it is the minimal solution of a fixpoint algebraic system over $\Rsemiring'$ .  
   \end{proposition}
   \begin{proof}
   Let $A\subseteq X$ be finite, stable and polynomial as above.
 Take the (finite) system of equations: $(r_a = s^*_a(r_a, x_Y))_{a \in A}$, where $s^*_a=(s_a)\vert_{A+Y}$, and let $\bar r_A$ be one of its solutions. Then, define $(t_x)_{x \in X}$ to be $\bar r_x$ if $x \in A$ and $0$ otherwise. 
Then for all $x\in X$, $s_x(t_x,x_Y)=0$: if $x\in A$, since $t_X=\bar r_{A} \sqcup 0_{X-A}$, we have
$s_x(t_x,x_Y)=s^*_x(\bar r_x,x_Y)=0$; if $x \not \in A$, then $(s_x)\vert_{A+Y}=0$, which again implies
%
		$s_x(t_X, x_Y)=0$.
   \end{proof}
%

  Combining Proposition  $\ref{prop:fintoalg}$ and Theorem  \ref{th:elimination},
  we also get:
   \begin{corollary}
   	\label{cor:propertoalg}

   		If $s_X\in \fps{\Rsemiring'}{X+Y}^X$ is a \textit{proper} finitary family, then for each $x \in X$  $(\fix s_X)_x$ is an algebraic power series. Moreover, for all but finitely many $x$, $(\fix s_X)_x=0$
   \end{corollary}

  \begin{example}
    Let $X=\N+1$, $Y=\{z\}$ and $s_X\in \fps{\Rinf}{X+Y}^X\equiv\fps{(\fps{\Rinf}{z})}{X}^{X}$ be the fps from Example \ref{ex:phors2}. 
  Then one can easily check that $W=\{1,\star\}\subset X$ is stable and polynomial, so the family is finitary. In other words, the fixpoint $a=\fix s_X$ can be computed by considering only the equations $a_1=s^F_1(a_1)$ (i.e.~\eqref{eq:alg1}) and $a_\gphors=a_{\star}=a_1$.
Hence $a_\gphors(z)\in \fps{\Rinf}{z}$ is algebraic over $\mathbb Q^+$ (with parameter $z$).

  \end{example}



\section{Finitely Bounded PHORS}


In this section we introduce a first class of PHORS, called $\Gfin$, whose corresponding  generating function is algebraic and whose AST and PAST problems are, consequently, decidable.

%

\subsection{Syntax of $\Gfin$ }

We use types defined by the grammar below:
\begin{align*}
\varphi,\psi&:= \1^n\mid \  !_k\varphi\multimap \varphi
 \qquad (1\leq n\in\mathbb N, \: k \in \mathbb N).
\end{align*}
We let $\tau$ be an abbreviation for the affine function type $!_1o\multimap o$. 
We define an order relation $\varphi\sqsubseteq \psi$ between types by induction by
\begin{align*}
\infer{\1^n\sqsubseteq\1^n}{}
\quad \infer{\xi\multimap\varphi\sqsubseteq\xi'\multimap\varphi'}{\xi\sqsubseteq \xi', \varphi\sqsubseteq\varphi'}
\quad
\infer{!_k\varphi\sqsubseteq !_h\psi}{k\leq h \quad\varphi\sqsubseteq\psi}
\end{align*}

The typing rules will use two kinds of contexts (as in \cite{DBLP:conf/lics/LiMO22}): 
\begin{varitemize}
\item non-linear contexts, noted $\nonterm,\nonterm'$ will be used for non-terminal symbols; these are finite sets of type bindings of the form $L:\varphi$; intuitively, we put no restrictions on the number of times a non-terminal is used;

\item graded contexts, noted $\Delta,\Delta'$ will be used for the arguments to be passed to the non-terminals; they are finite sets of type bindings of the form $x:_k \varphi$. Intuitively, the binding $x :_k \varphi$ means that we can use $x$ with type $\varphi$ \emph{at most} $k$ times.
\end{varitemize}
\noindent We define by induction the operations $\Delta+\Delta',k\Delta$ on graded contexts:
{\small
\begin{align*}
\Delta +\emptyset &= \Delta, \\
(\Delta,x:_{k}\varphi) + (\Delta',x:_{h}\varphi) &= (\Delta+\Delta'),x:_{k+h}\varphi;\\
k\emptyset&=\emptyset,\\
k(\Delta,x:_{h}\varphi)  &= k\Delta,x:_{kh}\varphi,
\end{align*}
}

\noindent where, as in e.g.~\cite{10.1145/3009837.3009890}, we are assuming that $\Delta$ and $\Delta'$ agree on each variable (i.e.~they contain the same type bindings but may disagree on the corresponding coefficients). 

A type judgement is an expression of the form $\nonterm\mid\Delta\vdash t:\varphi$.
The typing rules are in Fig.~\ref{fig:typerules}.
\begin{definition}[$\Gfin$]
A \emph{finitely graded PHORS} (noted $\Gfin$) is a triple $\gphors=(\nonterm, \C D, S)$, where $\nonterm$ is a finite set of typed non-terminals, $S\in\nonterm$ is the {start symbol} such that $\nonterm(S)=\1$, and $\C D$ associates each $L\in\nonterm$ with a derivation $\nonterm\mid\emptyset\vdash \lambda x_1.\dots. x_n.t:\varphi$
such that $\varphi\sqsubseteq \nonterm(L)$, $t$ contains no $\lambda$-abstraction and $\nonterm\mid x_1,\dots, x_k\vdash t:\1$.
\end{definition} 

Since affine implication is expressed by $!_1\varphi\multimap \psi$, 
it is immediate that all PAHORS are $\Gfin$, so all our results translate automatically to them. In particular, as the PHORS from Example \ref{ex:phors0} is affine, it is also a $\Gfin$.


%

%
%

\begin{figure}
\fbox{
\begin{minipage}{.7\textwidth}
\resizebox{.99\textwidth}{!}{
\begin{minipage}{1.1\textwidth}
\begin{align*}
	\begin{array}{ccc}
		\infer[ax1]{\nonterm\mid \Delta\vdash x:\varphi }{x:_p\varphi\in \Delta,\quad 1 \leq p} & &
		\infer[ax2]{\nonterm\mid\Delta \vdash L:\varphi }{ L:\varphi  \in \nonterm}\\[7pt]
%
		\infer[\lambda]{\nonterm\mid \Delta \vdash \lambda x. t : !_k \varphi \multimap \psi }{\nonterm \mid \Delta, x:_k \varphi \vdash t : \psi}
		&  &
		\infer[@]{\nonterm\mid \Delta + k\Delta' \vdash tu : \psi}
		{\nonterm\mid\Delta \vdash t :!_k\varphi \multimap \tau \qquad \nonterm\mid \Delta' \vdash u :\varphi }\\[9pt]
		\infer[\langle \rangle]{\nonterm\mid \Delta \vdash \langle t_1, \dots,t_n \rangle : \1^n}{\nonterm\mid \Delta \vdash t_i : \1 &  i=1,\dots, n}&  &
		\infer[\pi_i]{\nonterm\mid \Delta \vdash \pi_i t: \1}{\nonterm\mid \Delta \vdash t: \1^n &  i=1,\dots, n}
	\end{array}
	\end{align*}
	\begin{align*}
		 \infer[\oplus]{\nonterm\mid \Delta \vdash t \oplus t' : \1 }{\nonterm\mid \Delta \vdash t:\1 \quad \nonterm\mid \Delta \vdash t':\1} 
	\end{align*}
	\end{minipage}
	}
\end{minipage}
}
\caption{Typing rules of $\Gfin$.}
\label{fig:typerules}
\end{figure}

\begin{example}\label{ex:nonlin}
The $\gphors$ from \eqref{eq:phors1} in Section 2 is not affine, but is a $\Gfin$: 
letting $\nonterm(A)=\nonterm(B)=\tau$ and $\nonterm(H)=!_2\tau\multimap\tau$, we can correctly type  $\nonterm\mid  \emptyset \vdash \lambda f. \lambda x.t_H: \nonterm(H)$, where
$t_H= (H(A\circ f)x\oplus_a H(B\circ f)x)\oplus_a f(fx)$.

%
\end{example}

\begin{example}[a non-algebraic PHORS]\label{ex:nonalg}
The PHORS $\gphors$ below:
\begin{align*}
Lfx &= L(f\circ f) x\oplus_{p} fx \\
S&= L\ \mathrm{id}\  e
\end{align*}
is \emph{not} a $\Gfin$: if we try to type $t_L=L(y\circ y) x\oplus_{p} fx$
with $\nonterm(L)=!_k\tau\multimap \tau$, we obtain
$\nonterm \mid \emptyset\vdash \lambda y.\lambda x.t_L: \psi$, where $\psi=\ !_{2k}\tau\multimap \tau\not\sqsubseteq\nonterm(L)$. 
Indeed, the generating function $s^L(y_{\N},x)\in \fps{\Rsemiring}{\N+1}$ of $L$ satisfies 
$s^L(y_\N,x)=\frac{1}{2}(
y_1x+s^L( y^2_\N,x)
)$, which yields the solution $s^L(y_\N,x)=\sum_i\frac{1}{2^{i+1}}y^{2^i}x$, that is not algebraic (cf. Remark~\ref{rem:hadamard}). 
\end{example}

%

While the $\Gfin$ are not linear, they can be \emph{linearized}:
we can turn them into PAHORS having the same semantics and branch language. Notice, however, that the resulting affine PHORS might have size exponential in the original one.
More precisely, define the \emph{size} of a $\gphors$ as the number of its non-terminals times the maximum size of any of its terms, i.e.~
$\|\gphors\|=|\nonterm|\times\max\{|\C D(L)|\mid L\in\nonterm\}$, and let $\partial\gphors$ be the maximum grade occurring in $\gphors$. Then we have:

\begin{theorem}[Linearization]
	\label{theorem:linearisation}
For every $\Gfin$ $\gphors=(\nonterm, \C D, S)$ there exists a PAHORS $\AFF\gphors=(\nonterm, \C D',S)$ with size $\|\AFF\gphors\|\leq \partial\gphors{\|\gphors\|}$  such that 
$a_{\gphors}=a_{\AFF{\gphors}}$ and 
$\C L(\gphors)=\C L(\AFF\gphors)$.

\end{theorem}
\begin{proof}[Proof sketch]
Define $\AFF{\varphi}$ as $
\AFF{(o^n)}=o^n$ and 
$\AFF{(!_k\varphi\multimap \psi)}={\AFF{\varphi}\multimap\dots\multimap\AFF{\varphi}}\multimap\AFF\psi$, with $\AFF{\varphi}$ repeated $k$ times. 
Then, by induction, any derivation of $\nonterm\mid \Delta\vdash t:\varphi$ yields a derivation of $\AFF{\nonterm}\mid\Delta'\vdash t:\AFF \varphi$ in the affine system of \cite{DBLP:conf/lics/LiMO22}, where $|\Delta'|\leq \partial\gphors|\Delta|$: an order-$n$ $\Gfin$ yields then an order-$n$ PAHORS in which every equation $Lx_1\dots x_n$ translates into a new equation $Lz_1^1\dots z_{n}^{k_n}$, where, if $\nonterm(L)=!_{k_1}\varphi_1\multimap\dots \multimap !_{k_i}\varphi_i\multimap\dots\multimap o^n$, the variable $x_i$ is replaced by $k_i$ variables $z_i^1,\dots, z_{i}^{k_i}$.
\end{proof}

\begin{example}
The linearization of the $\Gfin$ from Example \ref{ex:nonlin} yields the PAHORS below:
{\small
\begin{align*}\label{eq:phors2}
Lf_1f_2x&=( L(A\circ f_1)(A\circ f_2)x \oplus_{a}
L(B\circ f_1)(B\circ f_2)x )\oplus_{a}f_1(f_2x)\\
Ax&= x\oplus_{b} \Omega\\
Bx&=x\oplus_{c} \Omega\\
S&=HIIe
\end{align*}
}
Notice that the unique functional variable $f$, that was used twice, is now replaced by \emph{two} functional variables $f_1,f_2$, used once.
\end{example}


\begin{example}
$\|\AFF\gphors\|$ may be exponential in $\|\gphors\|$:
take the order-2 $\Gfin$ $\gphors$ given by
{\small
\begin{align*}
F_1fx&= f(fx)\\
F_2f x&= F_1(f\circ f) x \oplus_{p} x\\
&\vdots\\
F_{n}fx&= F_{n-1}(f\circ f)x\oplus_{p}x\\
S&=F_n \ \mathrm{id} \  e
\end{align*}
}
As $\nonterm(F_n)=\ !_{2^{n}}\tau\multimap\tau$, 
$\|\AFF{\gphors}\| \leq \partial\gphors\|\gphors\|=
2^{\C O(\|\gphors\|)}\|\gphors\|$.
\end{example}

\subsection{From $\Gfin$ to FAS}

We will now show how, for any $\Gfin$ $\gphors$, the interpretation of the underlying PHORS yields a finitary FAS, whose minimal solution is thus an algebraic family.

First, from any $\Gfin$ we can canonically extract the underlying (simply typed) PHORS.
Observe first that the rules (and indeed the derivations) of $\Gfin$ can be seen as \emph{quantitative refinements} of the simply typed ones. Indeed, every type $\varphi$ is (uniquely) a \emph{refinement} of some simple type $\sat{\varphi}$: the simple types $\sat\varphi$ are defined by $\sat{(o^n)}=o^n$ and $\sat{(!_k\varphi\multimap\psi)}=\sat\varphi\to\sat\psi$.
Notice that $\varphi\sqsubseteq \psi$ implies $\sat\varphi=\sat\psi$. 
Moreover, for every derivation $\pi:\nonterm\mid\Delta \vdash t:\varphi$, 
by replacing each $\psi$ by $\sat\psi$, we obtain, by induction, a simply typed derivation $\sat\pi:\sat\nonterm,\sat\Delta\vdash t:\sat\varphi$.

 For any set $\Sigma$  and $k\in \N$, let $!_k\Sigma\subseteq !\Sigma$ be the set of multisets $\mu$ of maximal multeplicity $k$ (i.e $\max_{x \in \Sigma} \mu(x)\leq k)$. $!_k\Sigma$ corresponds to the monomials over $x_\Sigma$ of degree at most $k$ in each variable. 
  Define then, for every type $\varphi$, $\model{\varphi}$ as $\model{o^n}=1$ and $\model{!_k\varphi\multimap \psi}=!_k\model{\varphi}\times\model{\psi}$. We extend this to contexts by $\model{\emptyset}=0$ and $\model{\nonterm,x:\varphi}=\model{\nonterm}+\model{\varphi}$ (irrespectively of grades). Notice that $\model{\varphi}\subseteq\model{\sat\varphi}$, but the variable sets $\model{\varphi}$ are always finite.

   Our general idea is the following: given a $\Gfin$ $\gphors=(\nonterm, \C D, S)$, from the derivations $\C D(L)$ of $\nonterm\mid\emptyset \vdash t_L:\nonterm(L)$, we deduce a corresponding simply typed derivation of $ t_{\gphors}:\sat\nonterm\to \sat\nonterm$ yielding a family
   $\model{t}^{\Rsemiring}\in \fps{\Rsemiring}{\model{\sat\nonterm}}^{\model{\sat\nonterm}}$.  
  At the same time, we will show that the finite set $\model{\nonterm}\subset\model{\sat\nonterm}$ yields a stable and polynomial pair, and we will conclude that $\model{t_{\gphors}}$ is finitary, hence, by Corollary \ref{cor:propertoalg}, algebraic.    
   
   The fundamental ingredient is the following lemma:
   
  \begin{lemma}[stability lemma]\label{lemma:stability}
  Let $\Rsemiring=\fps{\Rinf}{\Sigma}$.
For every derivation $\nonterm\mid\Delta\vdash t:\varphi$, the pair $(\model{\nonterm}+\model{\Delta}, \model{\varphi})$ is stable and polynomial for $\model{t}^{\Rsemiring}\in \fps{\Rsemiring}{\model{\sat\nonterm}+\model{\sat\Delta}}^{\model{\sat\varphi}}$.

  \end{lemma}
  \begin{proof}[Proof sketch]
  Let a \emph{finitary} variable be any $\sigma\in\model{\nonterm},\model{\Delta},\model{\varphi}$, and a \emph{non-finitary} one be any $\sigma\in\model{\sat\nonterm}-\model{\nonterm},\model{\sat\Delta}-\model{\Delta},\model{\sat\varphi}-\model{\varphi}$.
By induction, one shows the following two claims:
   \begin{varenumerate} 
   \item for all non-finitary $\sigma$, $(\model{t}^\Rsemiring)_\sigma\big\vert_{\model{\nonterm}+\model{\Delta}}=0$;
   \item for all finitary $\sigma$, 
   $(\model{t}^\Rsemiring)_\sigma\big\vert_{\model{\nonterm}+\model{\Delta}}$ has 
   finite degree in each variable $x_{L:\nonterm(L)}$ and 
   degree $\leq k$ in each variable $x_{f:_k\psi}$.
   \end{varenumerate}
   \vskip-4mm
   \end{proof}
An immediate consequence of Lemma \ref{lemma:stability} is the following:
%
%

\begin{proposition}
	\label{prop:finitefamily}
For every $\Gfin$ $\gphors=(\nonterm,\C D,S)$, the fps
$\model{t_{\gphors}}^{\Rsemiring}\in \fps{\Rsemiring}{\model{\sat\nonterm}}^{\model{\sat\nonterm}}$ is finitary (via $\model\nonterm$).
\end{proposition}

Thanks to Proposition \ref{prop:fintoalg} and the observation that all coefficients of $\model{t_\gphors}^{\fpp{\Rinf}{z}}$ belong indeed to $\fpp{\Q}{z}$, we obtain then:
\begin{theorem}
	\label{Gfin-algebraic-theorem}
	For all $\Gfin$ $\gphors=(\nonterm, \C D, S)$, $(a_L)_{L \in \nonterm}$ is a $\Q^{+}$ fixpoint algebraic family with parameter $z$
\end{theorem}
\begin{corollary}
	\label{Gfin-algebraic-corollary}
	 For all $\Gfin$ $\gphors=(\nonterm, \C D, S)$,
	\begin{varenumerate}
\item for each $L \in \nonterm$, $a_L(z)$ is an algebraic power series.  
\item $\mathbb P(\gphors\downarrow)$ is a $\mathbb Q^{+}$ algebraic number and 
 $\sum_i\mathbb P(\gphors\downarrow_i)z^i$ is a ${\mathbb Q^{+}}$-algebraic power series.
\end{varenumerate}
\end{corollary}
As a consequence of Remark $\ref{remark:semilinear}$ we also obtain:
\begin{corollary}
	For all $\Gfin$ $\gphors=(\nonterm, \C D, S)$, the set $\{n \mid \:  S \text{ terminates in } n \text{ steps}  \}$
	is semilinear
\end{corollary}
\begin{remark}
	By interpreting  $t_L \oplus t_R$ as $x \model{t_L} + \overline x \model{t_L}$ (cf.~Remark \ref{rem:tropical}), the previous corollary could be strengthened to say that the set
	$\{(n,m) \mid \:  S \text{ terminates with } n \text{ left steps and } m \text{ rigth steps}\}$ is also semilinear.
	This could also been deduced by the fact that the branch language of a linear HORS is \emph{multiple context-free} \cite{DBLP:conf/mfcs/ClairambaultM19}, and languages of this class are known to be semilinear \cite{DBLP:conf/acl/Vijay-ShankerWJ87}, but it is still worth noticing that we reproved this fact via generating functions.
\end{remark}
%
%
%
%
%
%
%

We conclude by deducing the decidability of AST and PAST:
\begin{theorem}
\label{th:decidability}
	AST and PAST are decidable for any $\Gfin$.	
\end{theorem}
\begin{proof}[Proof sketch]
The  proof of Theorem \ref{Gfin-algebraic-theorem} is constructive: it gives a way to build effectively a FAS $(\vec x = p_i(\vec x))_{1 \leq i \leq n}$ that has $(a_L(z))_{L \in \nonterm}$ as its (minimal) solution.
We can then check if $\mathbb P(\gphors\downarrow) = a_\gphors(1)=1$ via the first order Existential Theory of the Reals $\Etor$ \cite{DBLP:journals/jacm/EtessamiY09}: first, we create a first order formula $\phi( \vec x)$ expressing the fact that $\vec x$ satisfies $\vec x = p_i(\vec x)$, then we can express the fact that $x_1=1$ by the formula $\exists \vec{x}(x_1=1 \land \phi(\vec x) \land \forall \vec y ((\vec y < \vec x) \implies \neg \phi(\vec y)  ))$.

As for PAST, remember that PAST implies AST. Then, given a $\gphors=(\nonterm, \C D, S) \in \Gfin$, we can first test if it is AST: if it is not, then it is not PAST. If it is AST, notice that, by Corollary \ref{Gfin-algebraic-corollary}, $a_\gphors(z)$ is an algebraic power series. Then, we can, as pointed out in Remark \ref{derivative-rational}, effectively find a rational function $r(z, y)$ such that $a'_\gphors(z)= r(z, a_{\gphors}(z)) $. But then $\mathbb E(\gphors\downarrow) =a'_\gphors(1)= r(1,1)$.
\end{proof}
%
%
%
%
%
%
%
%
%
%
%
\begin{remark}
The complexity of our decision procedure for AST and PAST can be reconstructed 
as follows:
for an order $n$ PHORS with non-terminals of maximal grade $k$ and maximal arity $r$, the set $\model{\nonterm}$ has at most $\exptower{k}{n} \cdot r^n $ elements, with $\exptower{a}{b}$ the tower of exponential of base $a$ and height $b$. As $r\leq\|\gphors\|$, 
%
%
the polynomial $\model{t_L}^{\Q^+}$ involves at most $\exptower{k}{n} \| \gphors\|^n$ variables and they can be computed by at most $|t_L|$ products or sum of polynomials: the total cost of these operations will be $O(\exptower{k}{n}^2 \| \gphors \|^{2n+1})$ and we have to perform them for each non terminal: the total time complexity of computing the FAS system is then $O(\exptower{k}{n}^2 \| \gphors \|^{2n+2})$. Using a polynomial space algorithm to decide $\Etor$ \cite{DBLP:journals/jacm/EtessamiY09}, we can then see that the space complexity of deciding AST with our method \emph{for fixed order PHORS} is polynomial  in their size $\|\gphors \|$, but \emph{grows superexponentially} when the order is not fixed. Since the algorithm to compute the rational function of Remark \ref{derivative-rational} is polynomial \cite{CHUDNOVSKY1986271}, we can conclude the same for PAST.
%
%
%
\end{remark}
\begin{remark}
Beyond AST/PAST, the fact that $a_{\gphors}(z)=\sum^{\infty}_{i=0} \mathbb P(\gphors\downarrow_i)z^i$ is algebraic 
can be used to make asymptotic estimates on the  coefficients $\mathbb P(\gphors\downarrow_i)$ for $i \to + \infty$ 
(cf.~Remark \ref{coefficients-rational}).
\end{remark}

\section{Bounded PHORS}

We introduce a second class of PHORS that extends the $\Gfin$ by admitting infinite grades $!_\infty$, corresponding to parameters that may be used an \emph{unbounded} number of times. In this way, thanks to the composability of the underlying algebraic power series, the $\Ginf$ are proved to be closed under \emph{composition}.

\subsection{Syntax of $\Ginf$}

While the type discipline of $\Gfin$ ensures that all variables are used in a bounded way,
the variable $z$ in $a_\gphors(z)$, which counts the number of probabilistic choices (e.g.~think again of the HORS \eqref{eq:introphors2}), can be thought of as an input that is used \emph{unboundedly} during reductions, without compromising algebraicity.

The key insight here is that we are describing PHORS via families of power series with coefficients taken in $\fps{\Rinf}{z}$: in other words, the variable $z$ is treated as a \emph{formal parameter} all along the way.
Our goal is then to capture those situations in which a variable is used as a parameter and can, thus, be duplicated unrestrictedly.

\begin{example}\label{ex:phorsLfx}
Consider the following PHORS $\gphors$:
\begin{align*}
Lfgx&= f(Lfg(Lfgx))\oplus_{a}gx\\
Ax&=x\oplus_b \Omega\\
Bx&=x\oplus_c \Omega\\
S&=LAB e
\end{align*}
Define the parenthesis symbols ``$($'' and ``$)$'' as, respectively, $ab$ and $ac$; then the branch language $\C L(\gphors)$ is formed by all words of the form $w)\in \{(,)\}^*$, where $w$ is ``well-parenthesized'', like e.g.~$w=(()(()))$, $w=(()(())())$, i.e.~it essentially corresponds to the \emph{Dyck language}.
Notice that $\gphors$ is not a $\Gfin$: 
if we try to type the term $t_L=\lambda fgx. f(Lfg(Lfgx))\oplus_{a}gx$ with finite grades, setting
$\nonterm(L)=!_{m}\tau\multimap!_n\tau\multimap\tau$, we obtain
$\nonterm\mid \emptyset\vdash t_L: \varphi$, where $\varphi= !_{2m+1}\tau\multimap !_{2n}\tau\multimap \tau\not\sqsubseteq\nonterm(L)$.
The only way out is to assign infinite grades $n,m=\infty$, i.e.~to let 
$\nonterm(L)=!_\infty \tau\multimap !_\infty\tau\multimap\tau$.

The interpretation of $L$ will be some power series $a_L(z,f,g,x)=\sum_{ijk}L_{ijk}f^ig^jxz^k$.
Now, observe that $a_L(z,f,g,x)$ can be rewritten as a power series $a_L(z,f,g)(x)$ with coefficients in $\fps{\Rinf}{z,f,g}$, that is a solution (letting e.g.~$a=\frac{1}{2}$) of the algebraic equation  
\begin{equation}\label{eq:catalan2}
a_L(z,f,g)=\frac{zf a_L(z,f,g)^2+zg}{2}.
\end{equation}
This equation is reminiscent of \eqref{eq:alg1} and can indeed be solved in a similar way, yielding $a_L(z,f,g,x)=\sum_{i=0}^{\infty}\frac{C_i}{2^{2i+1}}f^ig^{i+1}xz^{2i+1}$: there are 
$C_i$ well-parenthesized words of length $2i+1$, each using $($ $i$ times and $)$ $i+1$ times, each of them produced by a reduction making $2i+1$ choices, and thus having probability $\frac{1}{2^{2i+1}}$.
\end{example}

We now introduce a type discipline that captures situations like the one just sketched.
We use two different kinds of types:
\begin{align*}
\varphi,\psi&:= o^n\mid !_k\varphi\multimap \varphi\qquad (1\leq n\in \N, k\in \N)\\
\Phi,\Psi&:=\varphi\mid !_k\varphi\multimap \Phi\qquad (k\in \N\cup\{\infty\})
\end{align*}
The types $\varphi,\psi$ are called \emph{finitary} and are precisely those of $\Gfin$; the types $\Phi,\Psi$ are called \emph{infinitary}. The order relation on infinitary type $\Phi\sqsubseteq\Psi$ is inherited from the one for infinite types together with the new rule 
$\varphi\sqsubseteq\varphi',\Psi\sqsubseteq\Psi'\Rightarrow !_k\varphi\multimap \Psi\sqsubseteq !_\infty\varphi'\multimap\Psi'$

The typing rules will use three kinds of contexts:
\begin{varitemize}
\item the non-linear context $\nonterm$, made of infinitary bindings of the form $L:\Phi$;
\item a new parameters context $\C P$, made of finitary bindings of the form $w:_\infty \varphi$;
\item a graded context $\Delta$, made of finitary bindings of the form $x:_k \varphi$, with $k<\infty$.
\end{varitemize}
The rules are illustrated in Fig.~\ref{fig:typerulesinf}.
The rules for finitary types are as in the previous section; the new rules are
ax$_{\C P}$ and, notably, the infinitary abstraction and application rules $\lambda_\infty$ and $@_\infty$, which states that an infinitary function $t:!_\infty\varphi\multimap\Psi$ can only be applied to either a non-terminal or a parameter. 
Recall that in a PHORS each non-terminal $L$ is equated with a corresponding term $t_L$: so, when we apply a function $t:!_\infty \varphi\to \Psi$ to the symbol $L$, we are actually applying $t$ to the corresponding term $t_L$.

\begin{figure}
\fbox{
\begin{minipage}{.7\textwidth}
\resizebox{\linewidth}{!}{
\begin{minipage}{\linewidth}
\begin{align*}
		\infer[ax_\Delta]{\nonterm\mid\C P\mid \Delta\vdash x:\varphi }{x:_p\varphi\in \Delta,\quad 1 \leq p} \qquad
		\infer[ax_{\C P}]{\nonterm\mid\C P\mid\Delta \vdash w:\varphi }{ w:_\infty\varphi  \in \C P}
		 \qquad
		\infer[ax_\nonterm]{\nonterm\mid\C P\mid\Delta \vdash L:\Phi }{ L:\Phi  \in \nonterm}
\end{align*}
\begin{align*}
		\infer[\lambda]{\nonterm\mid\C P\mid \Delta \vdash \lambda x. t : !_k \varphi \multimap \Psi }{\nonterm\mid\C P \mid \Delta, x:_k \varphi \vdash t : \Psi}
		&  &
		\infer[@]{\nonterm\mid\C P\mid \Delta + k\Delta' \vdash tu : \Psi}
		{\nonterm\mid\C P\mid\Delta \vdash t :!_k\varphi \multimap \Psi \qquad \nonterm\mid\C P\mid \Delta' \vdash u :\varphi }
\end{align*}
\begin{align*}
		\infer[\lambda_\infty]{\nonterm\mid\C P\mid \emptyset \vdash \lambda w. t : !_\infty \varphi \multimap \Psi }{\nonterm\mid\C P, w:_\infty \varphi \mid \emptyset \vdash t : \Psi}
		&  &
		\infer[@_\infty]{\nonterm\cup\nonterm'\mid\C P\mid \Delta \vdash tu : \Psi}
		{\nonterm\mid\C P\mid\Delta\vdash t :!_\infty\varphi \multimap \Psi \qquad u:\varphi\in \nonterm'\cup\C P}
\end{align*}
\begin{align*}
		\infer[\langle \rangle]{\nonterm\mid\C P\mid \Delta \vdash \langle t_1, \dots,t_n \rangle : \1^n}{\nonterm\mid\C P\mid \Delta \vdash t_i : \1 &  i=1,\dots, n}&  &
		\infer[\pi_i]{\nonterm\mid\C P\mid \Delta \vdash \pi_i t: \1}{\nonterm\mid\C P\mid \Delta \vdash t: \1^n &  i=1,\dots, n}
	\end{align*}
	\begin{align*}
		 \infer[\oplus]{\nonterm\mid\C P\mid \Delta \vdash t \oplus t' : \1 }{\nonterm\mid\C P\mid \Delta \vdash t:\1 \quad \nonterm\mid\C P\mid \Delta \vdash t':\1} 
	\end{align*}
	\end{minipage}
	}
\end{minipage}
}
\caption{Typing rules of $\Gfin$.}
\label{fig:typerulesinf}
\end{figure}

\begin{definition}[$\Ginf$]
A \emph{(parametric) bounded PHORS} (noted $\Ginf$) is a tuple $\gphors=(\nonterm, \C P,\C D, S)$, where $\nonterm$ is a finite set of (infinitarily) typed non-terminals $L:\Phi$, $\C P$ is a finite set of (finitarily) typed parameters $w:_\infty\varphi$, $S\in\nonterm$ is the {start symbol} such that $\nonterm(S)=\1$, and $\C D$ is a function that associates each $L\in\nonterm$ with a derivation $\nonterm\mid\C P\mid\emptyset\vdash \lambda x_1.\dots. x_n.t:\Phi$ such that $\Phi\sqsubseteq \nonterm(L)$, $t$ contains no $\lambda$-abstraction and $\nonterm\mid x_1,\dots, x_k\vdash t:\1$. 

\end{definition} 

Notice that a $\Ginf$ $\gphors$ may contain \emph{open} parameters $w:_\infty\varphi\in \C P$. We call $ \gphors$ \emph{closed} when $\C P=\emptyset$.

The PHORS from Example \eqref{ex:phorsLfx} is a closed $\Ginf$, with $\nonterm(L)=!_\infty\tau\multimap!_\infty\tau\multimap\tau$. 
Instead, the non-algebraic PHORS from Example \eqref{ex:nonalg} is not a $\Ginf$:
while we can give type $\nonterm(L)=!_\infty \tau\to \tau$ to $L$, we cannot be apply it to the function $f\circ f:\tau$, as the latter is neither a parameter nor a non-terminal.

As for the $\Gfin$, we interpret each $\Ginf$ $\gphors=(\nonterm,\C P, \C D,S)$ via its
induced simply typed PHORS $(\sat\nonterm,\sat{\C P},\sat{\C D},S)$.
This yields then a family of fps
${\model{{\gphors}}}^{\Rsemiring}\in \fps{\Rsemiring}{\model{\sat\nonterm}+\model{\sat{\C P}}}^{\model{\sat{\nonterm}}}$, that we will look at as parametric on $\C P$, that is, as 
in $\fps{(\fps{\Rsemiring}{\model{\sat{\C P}}})}{\model{\sat\nonterm} }^{\model{\sat\nonterm}}$.
%

Two PHORS can be composed through their open parameters:
\begin{definition}[composition of $\Gfin$]
Given two $\Ginf$ $\gphors_1=(\nonterm_1,\C P\cup\{w:_\infty\nonterm(L)),\C D_1,S_1)$ and
$\gphors_2=(\nonterm_2,\C P,\C D_2,S_2)$, where $\nonterm_1\cap\nonterm_2=\emptyset$
and $L\in \nonterm_2$. Their \emph{composition through $L$} is the $\Ginf$ $\gphors_1 \circ_L\gphors_2=(\nonterm_1\cup\nonterm_2, \C P, \C D_1[L/w]\cup\C D_2,S_1)$.
\end{definition}

The composition of two PHORS is interpreted via the composition of the associated formal power series, as shown below:

\begin{lemma}\label{lemma:composition}
$\model{\gphors_1\circ_L\gphors_2}^{\Rsemiring}\left(w_{\model{\C P}}\right)=
\model{\gphors_1}^{\Rsemiring}\left(\model{\gphors_2}^{\Rsemiring}\left(w_{\model{\C P}}\right),w_{\model{\C P}}\right)$.
%
%
\end{lemma}

Composability enables a \emph{modular} analysis of the generating functions of a PHORS: we may for instance first focus on certain subsets of non-terminals, replacing all others with parameters, and only later put the pieces together via composition.

\subsection{From $\Ginf$ to FAS}

To show that the $\Ginf$ have algebraic generating functions, we will proceed in two steps.
First, we will consider finitary PHORS with parameters, which can be shown algebraic with the same approach as in Section 6. Then, we will show that any $\Ginf$ can be reduced to an equivalent parametric $\Gfin$, which will result in the algebraicity of its generating function.

\paragraph*{Finitary Parametric PHORS}
Consider a $\gphors=(\nonterm, \C P,\C D,S)$ that is parametric but \emph{finitary}, i.e.~all its non-terminal symbols have a finitary type. In other words, the derivations $\C D(L_i)$ are of the form $\nonterm\mid \C P\mid \emptyset \vdash t_{L_i}:\varphi$, where $\varphi\sqsubseteq \nonterm(L_i)$ and all types in $\nonterm$ and $\C P$ are finitary, but recall that the parameters $w$ in $\C P$ have infinite grade. 

By arguing as in the previous section, but reasoning over the semiring $\fps{\Rsemiring}{\model S}$, we can conclude that the pair of finite variable sets $(\model \nonterm,\model \nonterm)$ is stable and polynomial, yielding:

\begin{proposition}
For all finitary parametric $\Ginf$ $\gphors=(\nonterm, \C P,\C D, S)$, the fps $a_L(z,w_{\model{\C P}})\in \fps{\Rinf}{z,w_{\model{\C P}}}$ are all algebraic. 
\end{proposition}

\begin{example}\label{ex:Lfxpara}
Consider the finitary $\Ginf$, with two parameters $f,g:_\infty\tau$, obtained from the one in Example \ref{ex:phorsLfx} by introducing two new non-terminals
 $L_{[f,g]}:\tau$ and $S_{[f,g]}$ with the equations 
\begin{align*}
L_{[f,g]}x&= f(L_{[f,g]}(L_{[f,g]}x))\oplus_a gx\\
S_{[f,g]}&= L_{[f,g]}e
\end{align*}
The interpretation of $L_{[f,g]}:\tau$ is then the fps $a_{F}(z,f,g)(x)\in \fps{(\fps{\Rinf}{z,f,g})}{x}$ given by Equation \eqref{eq:catalan2}.
\end{example}

\paragraph*{Reducing $\Ginf$ to Parametric $\Gfin$}
Following the idea of Example \ref{ex:Lfxpara}, we show how any closed $\Ginf$ can be turned into a parametric $\Gfin$. 
The idea is to replace each non-terminal $L$ with possibly infinitely graded inputs $w_1,\dots, w_n$ with as many non-terminals $L_{[A_1,\dots,A_n]}$ as all possible ways of \emph{filling in} such inputs with other non-terminals from $\Ginf$. 

Let then $\gphors=(\nonterm,\emptyset, \C D,S)$ be a closed $\Ginf$. 
For any non-terminal $L$ whose type can be written as $\nonterm(L)=!_\infty \phi_1\multimap\dots\multimap !_\infty\phi_k\multimap \phi_L$,
let $\C P_L$ be the parameter context formed by all $w_i:_\infty\phi_i$, and let $T(L) \bydef \{[t_1, \dots t_n] \in (\nonterm\cup\C P_L)^* \: \mid \: \forall i \, (\nonterm \cup \mathcal{P}) (t_i) = \phi_i)  \}$ be the sets of \emph{sostituable terms}.
Let $\widetilde{\nonterm}$ be a new set of finitely typed non-terminals containing, for each $L\in \nonterm$ as above and $\gamma\in T(L)$, a non-terminal $\tilde{L}_\gamma:\phi_L$.
Finally, let $\C P$ be the concatenation of all $\C P_L$.

For each $L\in \nonterm$ and $\gamma\in T(L)$, define the following terms:
\begin{varitemize}
\item an applicative term $\nonterm\mid\C P\mid\emptyset \vdash \alpha_{L,\gamma}:\phi_L$ defined by 
$\alpha_{L,\gamma}:=L\gamma$;

\item a finitary term $\widetilde{\nonterm}\mid\C P \mid\emptyset\vdash \tilde{t}_{L,\gamma}:= \phi_L$ defined by 
$\tilde t_{L,\gamma}=( t_L\gamma)\big[ \tilde G_{\delta}/G\delta\big]_{G\in \nonterm, \delta\in T(G)}$.

\end{varitemize}

We have obtained in this way a finitary parametric PHORS $\gphors_{\mathrm{fin}}=(\widetilde\nonterm, \C P, \C D',S_{[]})$, and we have the equalities:
\begin{equation}\label{eq:alphaeq}
\alpha_{L,\gamma}\big[ t_L/  L] = t_L\gamma   =
\tilde{t}_{L,\gamma}\big[\alpha_{G,\delta}/\tilde{G}_\delta\big]_{G\in \nonterm,\delta\in T(G)}
\end{equation}

The following technical lemma relates then the minimal solutions of the systems of $\gphors$ and $\gphors_{\mathrm{fin}}$.

\begin{lemma}\label{lemma:fixpoint}
Given sets $\nonterm,\widetilde{\nonterm},\C P$ and morphisms
$t\in \Qrelkleisli{\Rsemiring}(\nonterm, \nonterm)$,
$\tilde t\in \Qrelkleisli{\Rsemiring}(\widetilde{\nonterm}+\C P, \widetilde\nonterm)$,
$\alpha\in \Qrelkleisli{\Rsemiring}(\nonterm+\C P, \widetilde{\nonterm})$,
if we have:
\begin{varenumerate}
\item $\alpha\circ \langle 0,1_\C P \rangle=0\in \Qrelkleisli{\Rsemiring}(\C P, \widetilde{\nonterm})$, 
\item $\alpha\circ \langle t\circ \pi_1,\pi_2  \rangle=
\tilde t\circ \langle \alpha\circ \langle \pi_1,\pi_2 \rangle,\pi_2\rangle\in \Qrelkleisli{\Rsemiring}(\nonterm+\C P, \widetilde{\nonterm})
$,
\end{varenumerate} 
then
$\alpha\circ \langle {\mathsf{fix}}_{0,\nonterm} t,1_{\C P}\rangle={\mathsf{fix}}_{\C P,\widetilde{\nonterm}} \tilde t\in \Qrelkleisli{\Rsemiring}(\C P, \widetilde{\nonterm})
$. 

\end{lemma}

Notice that condition 1.~of Lemma \ref{lemma:fixpoint} clearly holds for the interpretation of $\alpha=\langle \model{\alpha_{L,\gamma}}^{\Rsemiring}\rangle_{L,\gamma}$, and condition 2.~is precisely \eqref{eq:alphaeq}. 

\begin{theorem}
	\label{theorem:infalgebraic}
For any closed $\Ginf$ $\gphors$ there exists a parametric $\Gfin$ $\gphors_{\mathrm{fin}}$, $\|\gphors_{\mathrm{fin}}\|\leq \|\gphors\|^{|\nonterm|}$, such that $a_{\gphors}(z)=a_{\gphors_{\mathrm{fin}}}(z,1)$ and $\C L(\gphors)=\C L(\gphors_{\mathrm{fin}})$.
\end{theorem}
\begin{proof}[Proof sketch]
By inspecting the connected component formed by all non-terminals $\tilde L_{\gamma}$
``reachable'' from the source $S_{[]}$, one can check that none of these contains free parameters in $\gamma$. For the closed finitary PHORS $\gphors_{\mathrm{fin}}^*$ obtained by restricting to such non-terminals, Lemma \ref{lemma:fixpoint} yields then  
$\mathsf{fix}_{\nonterm} \model{\gphors}^{\Rsemiring}=\mathsf{fix}_{\widetilde\nonterm}\model{\gphors_{\mathrm{fin}}^*}^{\Rsemiring}$.
\end{proof}


We thus immediately obtain:
\begin{corollary}
For all closed $\Ginf=(\nonterm, \emptyset, \C D,S)$, and non-terminal $L\in \nonterm$, the  generating function $a_L(z)$ is algebraic. Moreover, the AST and PAST problems are decidable.
\end{corollary}
\begin{example}
If we apply this technique to the $\Ginf$ from Example \ref{ex:phorsLfx},
we must introduce all order-1 non-terminals $L_{[\gamma_1,\gamma_2]}$ with equations
$L_{[\gamma_1,\gamma_2]}x=\gamma_1(L_{[\gamma_1,\gamma_2]}(L_{[\gamma_1,\gamma_2]}x)\oplus_a \gamma_2x$, where $\gamma_1,\gamma_2$ are either non-terminals or parameters. With $\gamma_1,\gamma_2=f,g$, we precisely get the equation in Example \ref{ex:Lfxpara}. By restricting ourselves to the connected component of $S_{[]}$ we obtain then a \emph{closed} order-1 $\Gfin$ with non-terminals $S_{[]},A,B,L_{[A,B]}$.
\end{example}
\section{Related Work}

An extensive literature in analytic combinatorics has explored different kinds of combinatorial structures having an algebraic generating function, in particular with respect to their asymptotic behaviour (for a comprehensive treatment, see \cite{DBLP:books/daglib/0023751}).
In the setting of imperative programming, the idea of using generating functions to analyse probabilistic programs is developed in \cite{Klinkenberg2021} and \cite{Klinkenberg2024}; 
of particular interest, with respect to our work is the characterization of while loops giving rise to rational generating functions in \cite{Klinkenberg2021}.
On the side of formal language theory, there exists a huge amount of work about generating functions; particularly relevant to this work is the paper \cite{DBLP:journals/ita/AdamsFM13},  which studies the generating functions of Aho's indexed languages (corresponding to order-2 safe HORS, thus going beyond the usual framework of context-freeness) and they give an explicit procedure to compute these functions under some restrictions, but do not give a characterization of grammars with an algebraic generating function.

The connection between HORS and the relational semantics of linear logic first appears in  Melliès and Grellois' work \cite{DBLP:conf/csl/GrelloisM15}, \cite{DBLP:conf/mfcs/GrelloisM15}. Subsequently, in \cite{DBLP:journals/pacmpl/ClairambaultGM18} a linear-non linear typing system for HORS was introduced, allowing for a finer complexity analysis of model checking; its affine fragment is studied in \cite{DBLP:conf/mfcs/ClairambaultM19}. None of these works involve probabilistic behavior.

In \cite{DBLP:journals/lmcs/KobayashiLG20}, the problem of AST for PHORS is studied by directly translating a PHORS into (possibly higher-order) real functional equations: only in the case of order 1 PHORS, this produces a fixpoint algebraic system analogous the those considered in this work, thus implying a decidability result (while in the same work, undecidability since order-2 is proved). Techniques to approximate from above the probability of termination are also considered. In \cite{DBLP:conf/lics/LiMO22}, the decidability of AST/PAST for affine PHORS is proved. With respect to our approach, the translation from PHORS to real equations is defined via an intermediate automata-theoretic step, based on the equivalence between affine HORS and restricted pushdown automata proved in \cite{DBLP:conf/mfcs/ClairambaultM19}.

\section{Conclusion}

In this work we have shown that the combinatorial method of generating functions can be adapted, via the weighted relational semantics of linear logic, to the study of higher-order \emph{probabilistic} languages. 
We think that the main value of this work resides in opening the possibility of applying well-established methods from algebraic and analytic combinatorics to higher-order languages. We provided a first demonstration of this by showing how  the decidability of termination for certain classes of PHORS (proved in the literature with syntactic methods like automata theory or game semantics) can be re-established, and actually \emph{extended}, in a relatively elementary way, via the notion of algebraic power series.

At the same time, several other directions of application of combinatorial methods can be mentioned: on the one hand, could one capture a class of PHORS, extending the algebraic ones, giving rise to \emph{D-finite} power series \cite{DBLP:series/tmsc/KauersP11}, i.e.~fps defined by linear \emph{differential} equations, and could this be related to the recently explored connections between {higher-order differentiation} and fixpoints in the weighted relational model \cite{DBLP:conf/lics/GalalL24}?
On the other hand, could the vast body of work on \emph{asymptotic estimations} \cite{DBLP:books/daglib/0023751} of generating functions be applied to extract approximated  information about the probabilistic behavior of even larger classes of PHORS?
And, more broadly, could all these methods help shed light on 
the open problem of exactly placing the AST problem for PHORS within the arithmetic hierarchy?
%
%
%
%
%
%
%
%
%
%

\bibliography{refs} 
\newpage 

\appendix

\section{Proofs of Section 4}
We prove this by the connection between the relational model and non idempotent type systems.
Take a standard non idempotent intersection intersection type system, for example:
$$
\infer[app]{\biguplus \Gamma_i \uplus \Delta \vdash ts: \tau}{ \Gamma_i \vdash s: b_i \; (i \in \{1, \dots n\})\quad \Delta \vdash t:[b_1, \dots b_n] \to \tau}
\qquad
\infer[\lambda]{\Gamma\vdash \lambda x. t: \sigma \to \tau}{\Gamma, x: \sigma \vdash t:\tau}
$$
$$\infer[fix]{\biguplus \Gamma_i \uplus \Delta \vdash YM: \tau}{\Gamma_i \vdash YM: b_i \; (i \in \{1, \dots n\})\qquad \Delta \vdash M: [b_1, \dots, b_n] \to \tau}$$
Then we obtain:
\begin{lemma}
	\label{lambdaYlinear}
	Suppose that we have a $\lambda$Y term $t$ with $FV(t)=\{x_1, \dots x_n\}$ and that $x_1:\1 \dots x_n:\1 \vdash_{\lambda Y} t: \1 $. Then, for any derivation in this intersection type system with conclusion  $ x_1: \sigma_1 \dots  x_n: \sigma_n \vdash t:  \1 $, one of the following is true:
	\begin{enumerate}
		\item For every $i$, $\sigma_i = \emptyset$
		\item there exists one $i$ such that $\sigma_i=[\1]$ and for each $j \neq i$, $\sigma_j=\emptyset$
	\end{enumerate}
\end{lemma}
\begin{proof} 
	If we have a typing of $ x_1: \sigma_1 \dots  x_n: \sigma_n \vdash t:  \1 $, $t$ has an head normal form, then take any head normal form $h$ of $t$ (which must satisfy  $ x_1: \sigma_1 \dots  x_n: \sigma_n \vdash h:  \1 $, hence also $ x_1:\1 \dots x_n:\1\vdash h:\1$) . Then either $L = x$ or $L=e$. In the first case, we can only derive $x_i:[1] \vdash x_i:1$, in the second we just have $x_i:\emptyset \vdash e:1$\\
\end{proof}
From this we can deduce:
\begin{proposition}\label{prop:affine}
	For all first-order terms $x_1:\1,\dots, x_n:\1\vdash_{\PLY}t:\1$, 
	$\model{t}^{\Rinf}(x_1,\dots, x_n)\in \fps{\Rinf}{x_1,\dots,x_n}$ is affine: there exists scalars $w_0,w_1,\dots, w_n\in\mathbb R_{\geq 0}$ such that
	\[
	\model{t}^{\Rsemiring}(x_1,\dots, x_n)
	=w_0+w_1x_1+\dots+w_nx_n,
	\]
	where $w_0=\mathbb P[t\to^* e]$ and
	$w_{i+1}=\mathbb P[t\to^* x_{i+1}]$.
\end{proposition}
\begin{proof}
	By the soundness of the weighted relational model with respect to intersection type (\cite{DBLP:journals/jacm/EhrhardPT18}, \cite{DBLP:journals/corr/abs-2501-15637}), the monomials $x_0^{i_0}, \dots x_n^{i_n}$ that appear with positive coefficient in $\model{t}^{\Rinf}(x_1,\dots, x_n)\in \fps{\Rinf}{x_1,\dots,x_n}$ correspond to derivations $x_1: [\1^{i_1}] \dots x_n: [\1^{i_n}] \vdash t: \1$. 
\end{proof}
\section{Proofs of Section 6}
\subsection{Proof of Theorem \ref{theorem:linearisation}}
 We now define in a precise way the notion of linearization of a $\Gfin$. We remind that, following the terminology of \cite{DBLP:conf/lics/LiMO22}, by 'linear' here we mean what in a strict sense would be 'affine'.  First we define the linearisation of a type:
\begin{align*}
	&\Aff(\1^n) =\1 \quad n \geq 1\\
	&\Aff({!_0 \phi \multimap \psi}) = \Aff(\phi) \multimap \Aff(\psi) \\  
	&\Aff({!_k \phi \multimap \psi}) = \underbrace{\Aff(\phi) \multimap \dots \multimap \Aff(\phi)}_{k \textit{ times}} \multimap \Aff(\psi)
\end{align*}
We need now to define the linearisation of a the type derivation. In particular, to each derivation $\pi = \nonterm \mid \Delta \vdash \phi$ in $\Gfin$ we will associate a linear derivation $\Aff(\pi)$ of the form $\nonterm' \mid \Aff(\Delta) \vdash \phi$, where  $\nonterm'$ is obtained by $\nonterm $ by replacing every type binding $L:\varphi $ with $\AFF L: \Aff(\varphi)$ and $\Aff(\Delta)$ is obtained by $\Delta$ by replacing every type binding $x:_k \psi, \: k \geq 2$ with $k$ distinct bindings $x^{(1)}:\Aff(\psi) \dots x^{(k)}:\Aff(\psi)$
\begin{itemize}
	\item For $\pi = \infer[ax1]{\nonterm\mid \Delta\vdash x:\varphi }{x:_p\varphi\in \Delta,\quad 1 \leq p}$, $\Aff(\pi) \bydef \pi$ 
	\item For  $\pi = \infer[ax2]{\nonterm\mid\Delta \vdash L:\varphi }{ L:\varphi  \in \nonterm}$, $\Aff(\pi) \bydef \infer[ax2]{\nonterm' \mid\Delta \vdash \AFF L:\Aff(\varphi) }{ \AFF L:\Aff(\varphi)  \in \nonterm'} $ 
	\item If we have that $\pi$ ends with an application $\nonterm\mid \Delta + k\Delta' \vdash tu : \psi$ with premises $\nonterm\mid\Delta \vdash t :!_k\varphi \multimap \psi$ and $\nonterm\mid \Delta' \vdash u :\varphi$, suppose that we get inductively linearisations $\nonterm'\mid \Aff(\Delta) \vdash \Aff(t): \Aff(\varphi) \multimap \dots \Aff(\varphi) \multimap \Aff(\psi))$ (where $\Aff(\varphi)$ is repeated $k$ times)  and  $\nonterm' \mid \Aff(\Delta') \vdash \Aff(u) :\Aff(\varphi)$. Now we create $k$ different copies $u_1, \dots u_k$ of $u$ by replacing its free variables $x_1, \dots x_n$ with variables $x^{(i)}_1, \dots x^{(i)}_n, 1 \leq i \leq k$ and we create in the same way $k$ different copies of the derivation $\nonterm' \mid \Aff(\Delta') \vdash \Aff(u) :\Aff(\varphi)$. Now by $k$ linear applications we can type $\Aff(t) \Aff(u_1) \dots \Aff(u_k)$ with type $\Aff(\psi)$; this is the linearisation of $tu$
	\item If we have $\pi$ ends with an abstraction $\nonterm\mid \Delta \vdash \lambda x. t : !_k \varphi \multimap \psi$ with premises $\nonterm \mid \Delta, x:_k \varphi \vdash t : \psi$, then take its linearisation $\nonterm \mid \Aff(\Delta),x^{(1)}:\Aff(\varphi) \dots x^{(k)}:\Aff(\varphi) \vdash \Aff(t) : \Aff(\psi)$ and; by $k$ abstraction we then obtain a linear derivation of $\lambda x^{(1)} \dots x^{(k)}.\Aff(t): \Aff(\phi) \multimap \dots \Aff(\phi) \multimap \Aff(\psi))$,  (where $\Aff(\phi)$ is repeated $k$ times).
	\item In the case of $\pi_i$ and $\langle \rangle$, we can just apply them to their linearised premises and obtain a linearised conclusion.
\end{itemize} 
Now given a $ \gphors=(\nonterm, \C R, S)$, we define its linearisation $\AFF\gphors=(\nonterm, \C R',\AFF S)$ by $\C R'(L) \bydef \Aff(\C R(L))$ for each non terminal $L$. In particular, $\AFF{S}$ will have equation $\AFF S= \Aff(\C R(S))$ Call $\AFF{\redbigp{ }}$ the corresponding rewriting\\
To see that $\AFF \gphors$ generates the same tree as $\gphors$, we need the following lemmas:
\begin{lemma}
	Let $\nonterm \mid \Delta, x:_k  \vdash  t[x]$ be a term typable in $\Gfin$ and let $\nonterm' \mid \Aff(\Delta), x^{(1)}: \Aff(\varphi), \dots x^{(k)}: \Aff(\varphi) \vdash  \Aff t[x^{(1)}, \dots x^{(k)}]$ be its linearisation. Let $u$ be a term of type $\varphi$ and let $\Aff(u)$ be its linearisation; then, up to a renaming of the free variables:
	$$\Aff(t[x/u]) = \Aff(t)[x^{(1)}/\Aff(u)_1, \dots x^{(k)}/\Aff(u)_k] $$ 
	where $\Aff(u)_1 \dots \Aff(u)_k$ are obtained from $\Aff(u)$ by replacing its free variables with fresh ones.
\end{lemma}
\begin{proof}
	By induction on the definition of linearisation.
\end{proof}
Then we have
\begin{lemma} For each non terminal $L$ and $d \in \{l,r\}$
	$L t_1 \dots t_n  \redbigp{\mathsf d,p} t_L[x_1/t_1, \dots x_n/t_n]$ iff $\Aff(L t_1 \dots t_n)  \redbigp{\mathsf d,p} \Aff(t_d[x_1/t_1, \dots x_n/t_n])$
\end{lemma}
\begin{proof}
	We have that $\Aff(Lt_1 \dots t_n)= \AFF L \Aff(t_1)^{n_1} \dots \Aff(t_k)^{n_k} \AFF{\redbigp{\mathsf d,p}} \Aff(t_d)[x_1^{n_1}/(t_1)^{n_1}, \dots x_1^{n_1}/(t_k)^{n_k}]$, but this by the previous lemma is equal (all terms being closed, no renaming of the free variables is needed) to $\Aff(t_d[x_1/t_1, \dots x_n/t_n])$
\end{proof}  
Finally, notice that $e$ is the only normal form both for $\gphors$ and $\AFF \gphors$ and that the linearisation of $e$ is $e$. From this we see that there exists a bisimulation between the rewriting system of $\gphors$ and the one of $\AFF \gphors$, hence they generate the same tree.
\subsection{Proof of Proposition \ref{prop:finitefamily}}
To prove proposition 6.3, we prove first the following correctness lemma:
\begin{lemma}[Correctness of $\Gfin$ for ground terms]
	\label{lemma:correctness}
	Given a term $\nonterm\mid\Delta\vdash t: \1$, let $s \bydef \model{t: \1}_{}^{\Rsemiring} $ be its $\Qrelkleisli \Rsemiring$interpretation, which consists of a single power series. Then, $s_\sigma\big\vert_{\model{\nonterm}+ \model{\Delta}}$ is a polynomial. Moreover, for each $f:_k \psi \in \Delta$, the total degree of $s$ in the variables $x_{f:_k \psi}$ is at most $k$
\end{lemma}
We recall that given a polynomial $ P = \sum_{\kappa \in !S} a_\kappa x_S^\kappa$, and a subset $S' \subset S$, the total degree of $P$ in the variables $x_{S'}$ is $\max_{\kappa \mid a_\kappa \neq 0} (\sum_{i \in S'} \mu(i))$.
We don't prove directly the last lemma, as the induction hypothesis would be not sufficient, but we pass through an intermediate step that reminds of reducibility candidates.
\begin{definition}
	For each type $\varphi$, we define by induction a set ${\Com}(\varphi)$ such that $\model{\varphi} \subseteq \Com(\phi) \subseteq \model{\sat{\varphi}}$:
	\begin{enumerate}
		\item $\Com(\1^n)= \model{\1^n}$
		\item $\Com( !_k \psi \multimap \varphi \tri S \to T) \bydef \{(\mu, f) \in !\model{S} \times \model{T}\mid \supp \mu \subseteq \Com(\psi) \implies (|\mu| \leq k \land \Com(f))\}$
	\end{enumerate}
When $\Delta= f:_{k_1} \psi_1, \dots f:_{k_n} \psi_n$ is a context, $\Com(\Delta) \bydef \{(p_1, \dots p_n)\mid \forall i \: p_i \in \Com(\psi_i)\}$.
\end{definition}
Now we can prove by induction the following lemma:
\begin{lemma}
	Given a term $\nonterm\mid\Delta\vdash t: \varphi$, let $s_{\model{\sat{\varphi}}} \bydef \model{t: \sat{\varphi}}^{\Rsemiring} $ be its $\Qrelkleisli \Rsemiring$interpretation. We have that:
	\begin{enumerate}
		\item For each $\sigma \in \model{\sat{\varphi}} \setminus \Com(\model{\varphi})$, 
		\begin{equation}
			s_\sigma\big\vert_{\model{\nonterm}+ \Com(\model{\Delta})}=0
		\end{equation}
		\item For each $\sigma \in \Com(\model{\varphi})$, $s_\sigma\big\vert_{\model{\nonterm}+ \Com(\model{\Delta})} \in \fpp{\Rsemiring}{\model{\nonterm}+ \Com(\model{\Delta})}$. Moreover, for each $f:_k \psi \in \Com(\Delta)$, the total degree of $s_\sigma$ in the variables $x_{f:_k \psi}$ is at most $k$
	\end{enumerate}
\end{lemma}
	\begin{proof}
		By induction on the type derivation $\nonterm\mid\Delta\vdash t:\varphi$. We immediately see this is true for the axioms. We discuss here the crucial cases: application and abstraction.\\ 
		Application: suppose we know that the thesis is true for $\nonterm\mid\Delta\vdash t:!_k \psi \multimap \varphi$ and $\nonterm\mid\Delta'\vdash u: \psi$. We want to prove that it is true for $\nonterm\mid\Delta+ k\Delta'\vdash (tu): \sat{\varphi}$ Let $s_{\model{\sat{\varphi}}} \bydef \model{tu: \sat{\varphi}}^\Rsemiring$ We know that by definition:
		$$s_\sigma\big\vert_{\model{\nonterm}+ \Com(\model{\Delta}+ \model{\Delta'})} = \left( \sum_{\mu \in !{\model{T}}} \model{t}^\Rsemiring_{\mu, \sigma} (\model{u}^\Rsemiring)^\mu \right) \Big\vert_{\model{\nonterm}+ \Com(\model{\Delta}+ \model{\Delta'})} $$

	When we restrict the RHS to variables from $\model{\nonterm}+ \model{\Delta}+ \model{\Delta'}$, we obtain:
	\begin{equation}
		\label{eq:sumrestrict}
	\sum_{\mu \in !{\model{T}}} \model{t}^\Rsemiring_{\mu, \sigma} \Big\vert_{\model{\nonterm}+ \Com(\model{\Delta})} (\model{u}^\Rsemiring)^\mu \Big\vert_{\model{\nonterm}+ \Com(\model{\Delta'})}
	\end{equation}
	If $\sigma \in \model{T} \setminus \Com(\model{\sigma})$, then we have two cases:
	\begin{enumerate}
		\item If $\supp \mu \not \subseteq \Com(\psi)$, then in the product $(\model{u}^\Rsemiring)^\mu$ there is a factor $(\model{u}^\Rsemiring)_m$ with $m \not \in  \Com(\psi)$; its restriction is then $0$ by inductive hypothesis and so the whole product is $0$
		\item If $\supp \mu \subseteq \Com(\psi)$, then $(\mu, \sigma) \not \in \Com(!_k \psi \multimap \varphi)$. Hence, by inductive hypothesis $\model{t}^\Rsemiring_{\mu, \sigma}=0$ and the product is $0$
	\end{enumerate}
	So we deduce that the restriction of $s_\sigma$ vanishes.\\
	Now we prove part (2): let $\sigma \in \model{!_k\psi \multimap \varphi}$. By part (1) of the inductive hypothesis, for all $\mu$ such that  $|\mu|>k$ and $(\model{u}^\Rsemiring)^\mu \Big\vert_{\model{\nonterm}+ \Com(\model{\Delta})} \neq 0$  we have that $\model{t}^\Rsemiring_{\mu, \sigma} \Big\vert_{\model{\nonterm}+ \Com(\model{\Delta})}=0$, hence \ref{eq:sumrestrict} is a finite sum, each term being a polynomial by induction hypothesis (2); we conclude that the sum is a polynomial. Now, let $f:_l \xi \in \Delta + k\Delta'$; assume that $f:_n \xi {\Delta}$ and $f_m: \xi {\Delta'}$ ($l= n+km$) and let $\vec y$ be the set of variables $\Com(\model{{f:_l \xi}})$. We have that the total degree in $\vec y$ of each $\model{t}^\Rsemiring_{\mu, \sigma} \Big\vert_{\model{\nonterm}+ \Com(\model{\Delta})}$ that gives a contribution to the sum is (by inductive hypothesis) at most $n$ and the total degree in $\vec y$ of each   $(\model{u}^\Rsemiring)_{j}\Big\vert_{\model{\nonterm}+ \Com(\model{\Delta})}, \: j \in \model{\sat{\psi}}$  that gives a contribution to the sum is at most $m$. Since for all $\mu$ such $|\mu|>k$ we have  that $\model{t}^\Rsemiring_{\mu, \sigma} \Big\vert_{\model{\nonterm}+ \model{\Delta}}=0$, the maximal total degree of a non zero product $\model{t}^\Rsemiring_{\mu, \sigma} \Big\vert_{\model{\nonterm}+ \Com(\model{\Delta})} (\model{u}^\Rsemiring)^\mu \Big\vert_{\model{\nonterm}+ \Com(\model{\Delta'})}$ is $n+km$.\\
	Abstraction: suppose we know that the thesis is true for $\nonterm\mid\Delta, x_k: \psi \vdash t: \varphi$. We want to prove that it is true for $\nonterm\mid\Delta \vdash \lambda x. t: !_k \psi \multimap \varphi$. If we write $\model{t: \sat{\varphi}}^\Rsemiring_t$, $t \in \model{\sat{\varphi}}$ as a fps $s_t=\sum_{\kappa \in !{\sat{\varphi}}} r_{\kappa,t} (x_{{x:\sat{\varphi}}})^\kappa$ (with $r_\kappa \in \fps{\Rsemiring}{\nonterm+\Delta} $), we know by definition that 
	$$\model{ \lambda x. t:  \sat{\psi} \to \sat{\varphi} }^\Rsemiring_{\mu, t} \Big\vert_{\model{\nonterm}+ \Com(\model{\Delta})} = r_{\mu, t} \Big\vert_{\model{\nonterm}+ \Com(\model{\Delta})} $$
	Take $(\mu, t) \not \in \Com(!_k\model{\psi} \multimap \model{\varphi})$. This means that $\supp \mu \subseteq \Com(\psi)$, but either $t \not \in \Com(\phi)$ or $|\mu| > k$. In the first case, we know $s_t=0$. Since $\supp \mu \subseteq \Com(\phi)$, then $x_{x:S}\vert_{\model{\nonterm}+ \Com(\model{\Delta})^\mu \neq 0}$ (as none of its variables with positive degree is equated to $0$ by this restriction), and so we must have also $r_{\mu, t} \Big\vert_{\model{\nonterm}+ \Com(\model{\Delta})}=0$.  
	Now assume that $t \in \Com(\phi)$ and $\|\mu| > k$. By inductive hypothesis, we know that $s\vert_{\model{\nonterm}+ \Com(\model{\Delta})}$ is a polynomial of total degree at most $k$ in the variables $x_{x:S}$: hence $r_{\mu, t} \Big\vert_{\model{\nonterm}+ \model{\Delta}}=0$ for $\supp \mu \subseteq \Com(\psi)$ and $|\mu| > k$.
	 $r_{\mu, t} \Big\vert_{\model{\nonterm}+ \Com(\model{\Delta})}=0$. Point (2) is in this case obvious. 
\end{proof} 
From the previous lemma, lemma \ref{lemma:correctness} is obvious by taking $\varphi=\1$  and observing that if $s\Big\vert_{\model{\nonterm}+ \Com(\model{\Delta})}=0$, then  $s\Big\vert_{\model{\nonterm}+ \model{\Delta}}=0$ (as $\Delta \subseteq \Com(\Delta)$)
From this we obtain the following lemma about terms $\lambda x_1 \dots x_n.t: !_{k_1} \psi_1 \multimap \dots !_{k_n} \psi_n \multimap \1 $:
\begin{lemma}
	Given a term $$\nonterm\mid\emptyset \vdash : \lambda x_1 \dots x_n.t: \varphi$$ with $\varphi = !_{k_1} \psi_1 \multimap \dots !_{k_n} \psi_n \multimap \1$, let $s_{\model{\sat{\varphi}}} \bydef \model{t: \sat{\varphi}}^{\Rsemiring} $ be its $\Qrelkleisli \Rsemiring$interpretation. Then:
	\begin{enumerate}
		\item If $\sigma \not \in \model{\sat{\varphi}} \setminus \model{\varphi}$, then $s_\sigma\Big\vert_{\model{\nonterm}}=0$
		\item If  $\sigma \in \model{\sat{\varphi}}$, then $s \Big\vert_{\model{\nonterm}}$ is a polynomial. 
	\end{enumerate}
\end{lemma}
\begin{proof}
	By lemma \ref{lemma:correctness}, the interpretation of $\nonterm\mid x_1:_{k_1} \dots x_n:_{k_n} \vdash : t: \1$ is a polynomial of total degree at most $k_i$ in the variables $x_{x_i:_{k_i}: \psi_i}$
\end{proof}
\section{Proofs of Section 7}
\subsection{Proof of Lemma \ref{lemma:composition}}
Here, we prove that what has been defined in Definition 7.2 as the composition of two $\Ginf$ really deserves this name.
Let  $\Ginf$ $\gphors=(\nonterm_1,\C P ,\C R_1,S_1)$ be a $\Ginf$. Recall that its $\Qrelkleisli{\Rsemiring}$ interpretation $\model{\gphors}$ is then an arrow in $\Qrelkleisli{\Rsemiring}(\model{P}, \model{N})$, which is a family $s_{\model{\nonterm}} \in \fps{\Rsemiring}{\model{\C P}}$ of fps.
\begin{lemma}
	Let  $\Ginf$ $\gphors_1=(\nonterm_1,\C P\cup\{w:_\infty\varphi),\C R_1,S_1)$ and
	$\gphors_2=(\nonterm_2,\C P,\C R_2,S_2)$ be two $\Ginf$ with $\nonterm_1 \cap \nonterm_2= \emptyset$ and $\nonterm_2(L) = \phi$. Let $\gphors_1(\gphors_2)$ be their composition along $L$. Then we have for each $L_1 \in \nonterm_1 $,
	$$\model{\gphors_1(\gphors_2)}_{L_1}^\Rsemiring(y_{\model{\C P}}) = \model{\gphors_1}^\Rsemiring_{\model{L_1}}(\model{\gphors_2}^\Rsemiring_{\model{L}}(y_{\model{\C P}}),y_{\model{\C P}} )$$
	and for each $L_2 \in \nonterm_2$,
	$$ \model{\gphors_1(\gphors_2)}_{L_2}^\Rsemiring(y_{\model{\C P}}) = \model{\gphors_2}_{L_2}^\Rsemiring(y_{\model{\C P}}) $$
\end{lemma}
\begin{proof}
	Let let respectively $s_{\model{\nonterm_1}}(y_\phi, y_{{\C P}})$, $r_{\model{\nonterm_2}}(y_{{\C P}})$ be the interpretations of $\gphors_1$ and $\gphors_2$. Then $s_{\model{\nonterm_1}}(y_\phi, y_{\nonterm{\C P}})$ is the minimal solution of the system $(S_1)$ of fixpoint equation:
	\begin{equation}
		\label{eq:fixg1}
		z_{\model{\nonterm_1}}= p_{\model{\nonterm_1}}(z_{\model{\nonterm_1}}, z_w,y_{\C P} )
	\end{equation}
	and  $r_{\model{\nonterm_2}}(y_{{\C P}})$  is the minimal solution of the system of fixpoint equations 
	\begin{equation}
		\label{eq:fixg2}
		z_{\model{\nonterm_2}}= q_{\model{\nonterm_2}}(z_{\model{\nonterm_1}}, y_{\C P} )
	\end{equation} 
	Then, $\model{\gphors_1(\gphors_2)}$ is the fixpoint solution of the system:
	\begin{equation}
		\label{eq:fixgcomp}
		\begin{cases}
			z_{\model{\nonterm_1}}= p_{\model{\nonterm_1}}(z_{\model{\nonterm_1}}, z_{\model{L: \phi}},y_{\C P} )\\
			z_{\model{\nonterm_2}}= q_{\model{\nonterm_2}}(z_{\model{\nonterm_1}}, y_{\C P} )
		\end{cases}
	\end{equation}
	Clearly, since $q_{\model{\nonterm_2}}$ do not involve any variable from other equations, the fixpoint iteration for \ref{eq:fixgcomp}, restricted to the variables $z_{\model{\nonterm_2}}$ will be equal to the fixpoint iterations for \ref{eq:fixg2}. Hence the minimal solution of \ref{eq:fixgcomp}, restricted to the variables $z_{\model{\nonterm_2}}$ will be  $r_{\model{\nonterm_2}}$. Let then $(t_{\model{\nonterm_1}},r_{\model{\nonterm_2}})$ be the full minimal solution of \ref{eq:fixgcomp}. We see that $t_{\model{\nonterm_1}}$ is then the a solution of $z_{\model{\nonterm_1}}= p_{\model{\nonterm_1}}(z_{\model{\nonterm_1}}, r_{\model{\nonterm_2}} ,y_{\C P} )$; it must be the minimal one, as if there were a solution $t'_{\model{\nonterm_1}} \leq t_{\model{\nonterm_1}}$, this would give a smaller solution of \ref{eq:fixgcomp}. This proves the lemma.
\end{proof}
\subsection{Proof of Theorem \ref{theorem:infalgebraic} }
We consider a parametric bounded PHORS $\gphors=(\nonterm, \C P,\C R, S)$ and we build from it a $\Gfin$ with parameters $R(\gphors)$ . In the case that $\gphors$ is closed, from $R(\Gfin)$ we can extract a subset of nonterminals that form a (parameter free) $\Gfin$. Up to isomorphism (from a fps-theoretic viewpoint, up to commutation of the variables),we can always work in the case then non terminals have type $ !_{\infty} \psi_1 \multimap \dots\multimap  !_{\infty} \psi_j \multimap \varphi$. 
\begin{definition}
	The set $T(L)$ of substituable terms for a non terminal $L \in \nonterm, \: \nonterm(L)= !_{\infty} \psi_1 \multimap \dots\multimap  !_{\infty} \psi_j \multimap \varphi$ is defined as follows:
	$$T(L) \bydef \{[t_1, \dots t_n] \in \nonterm^* \: \mid \: \forall i \, (t_i \in \nonterm \cup \mathcal{P}\land (\nonterm \cup \mathcal{P}) t_i = \psi_i)  \}  $$
\end{definition}
Now we define the enlarged set $\nonterm'$ of non-terminals of finite graded type, that will be used in $R(\gphors)$.
\begin{definition}[Reduced nonterminals]
	Let $\nonterm'$ be defined as:
	$$\bigcup_{L \in \nonterm}\{\tilde{L}_\gamma : \varphi \mid \nonterm(L)=!_{\infty} \psi_1 \multimap \dots\multimap  !_{\infty} \psi_j \multimap \varphi, \;  \gamma \in T(L)\} $$
\end{definition}

Now we define a translation from $\Ginf$ applicative terms of finite type (i.e terms $t$ that can be typed as $\nonterm\mid\C P\mid\Delta\vdash t: \varphi $) to $\Gfin$ terms $R(t)$ that can be typed as $\nonterm' \mid \Delta \vdash R(t): \phi$. 
First we prove the following:
\begin{lemma}
	If $\nonterm\mid\C P\mid\Delta \vdash t: \tau$ is an applicative term of type $!_{\infty} \psi_{1} \multimap \tau$, then it is of the form $L t_{1} \ldots t_{k}$, where $L \in \mathcal{N}$ and $t_{i} \in \mathcal{N} \cup \mathcal{P}$.
\end{lemma}
\begin{proof}
	The only way an infinitary type can appear in $\Ginf$ applicative term is because of the axiom of a non-terminal with infinitary type. The only rule we can apply after this kind of axiom is a sequence of $@_\infty$.
\end{proof}
Then we get: 
\begin{lemma}
	For every $\Ginf$ applicative term $\nonterm\mid\C P\mid\Delta\ \vdash t: \phi$ of finite type, there exists a $\Gfin$ applicative terms $\mathcal{N} \cup \mathcal{F}\mid \mathcal{P}\mid\Delta \vdash t': \phi$ such that $t=t^{\prime}\left[f_{1} / r_{1} \ldots f_{l} / r_{l}\right]$ and each $r_{i}$ is a $\Ginf$ applicative term of the form $ L t_{1} \ldots t_{k}$, where $L \in \nonterm$ and $t_{i} \in \nonterm \cup \mathcal{P}$.
\end{lemma}
\begin{proof} By induction, if $\mathcal{N}\mid \mathcal{P}\mid \Delta \vdash t: \tau$ is an axiom of finite type, $t^{\prime}=t$. If $\mathcal{N} \mid \mathcal{P} \mid \Delta \vdash L t_1 \dots t_k, \mathcal{N}(L)= !_{\infty} \psi_1 \multimap \ldots !_{\infty} \psi_k \multimap \varphi$, then $t=f_{1}, \mathcal{F}=\left\{f_{1}:  !_{\infty} \psi_1 \multimap \ldots !_{\infty} \psi_k \multimap \varphi\right\}$ and $r_{1}=L$.\\
If $t=t_{1} t_{2}$ and the last step was a finitary application, take the terms $t_{1}^{\prime}$ and $t_{2}^{\prime}$ whose existence is recursively ensured; we can assume that the free variables $f_{1}, \ldots f_{n}$ and $f_{1}^{\prime}, \ldots f_{n^{\prime}}^{\prime}$ and the contexts $\mathcal{F}$ and $\mathcal{F}^{\prime}$ are disjoint upon renaming and define $t=t_{1}^{\prime} t_{2}^{\prime}$.\\
\end{proof}
Now we define a translation from applicative terms typable in $\Ginf$ to terms typable in $\Gfin$: 
\begin{definition}
 Let $\nonterm\mid \mathcal{P}\mid \Delta \vdash t: \varphi$ be $\Ginf$ applicative term of finite type. Then we define its translation $R(t)$ to be the following $\Gfin$ applicative term that can be typed as $\mathcal{N}^{\prime} \mid \Delta \cup \mathcal{P} \vdash R(t): \varphi$. Given $t$, take $t^{\prime}$ (typable as $\mathcal{N} \cup \mathcal{F}\mid \mathcal{P}\mid\Delta \vdash t': \phi$) as defined in the last lemma. For each $r_{i}=L t_{1} \ldots t_{n}$ as in the last lemma, define $R(r_{i}):=\tilde{L}_{\left[t_{1} \ldots t_{n}\right]} \in \nonterm'$. Then define $R(t)$ to be the result of replacing every $f_{i}$ in $t^{\prime}$ by $R\left(r_{i}\right)$. 
\end{definition}
Now we are ready to define the translation from a closed $\Ginf$ $\gphors$ to a $\Gfin$ $R(\gphors)$ having the same value tree:
\begin{definition}
Now define a new PHORS $\mathcal{G}'$ with nonterminals $\N'$ as follows: for each $L \in \nonterm$ with equation $L \vec f = t_L$ and for each $\tilde L_\gamma \in \N'$ take the following equation:
$$\tilde L_\gamma x_1 \dots x_j = R(t_L[\vec g / \gamma]) $$
\end{definition}
This defines a finitary PHORS with parameters (as some of the substituable term in $\gamma$ might be parameters). 
\begin{definition}
	We define by induction the relation of dependency between non terminals: we say that a nonterminal $L$ depends over a nonterminal $G$ if either $G$ occurs in the rewriting rule of $L$ or there exists a non terminal $G'$  that occurs in the definition $L$ and $G'$ depends over $G$ 
\end{definition}
Let $L_{[p_1, \dots p_i]}$ be a reduced non terminal; we will say that it is open if at least one of the $p_i$ is a paramter. We will show now that in the recursive unfolding of a closed non terminals, only closed non terminals appear:
\begin{lemma}
	If $\gamma$ contains free parameter $p_1, \dots p_j$ and the rewriting rule of $\tilde L_{\gamma}$
	is $\tilde L_\gamma x_1 \dots x_j \redbigp{p, d}R(t_d[\vec g / \gamma]) $, then $R(t_L[\vec g / \gamma])$ only contains non-terminals $G_\gamma'$ where the parameters appearing in $\gamma'$ are amongst $p_1, \dots p_j$. In particular, if $L_\gamma$ is closed, all non terminals occuring in $R(t_d[\vec g / \gamma]) $ are closed.
\end{lemma}
\begin{proof}
	The only free variables in $t_d$ are paramters $g_1, \dots g_n$ and the $x_1, \dots x_j$. Since the only free parameters in $\gamma$ are $p_1, \dots p_j$, $FV(t_d[\vec g / \gamma]) \subseteq \{x_1, \dots x_n, p_1, \dots p_j\}$: it does contains only $p_1, \dots p_j$ as parameters, hence any reduced nonterminal in the translation $R(t_d[\vec g / \gamma])$ will only have this parameters. 
\end{proof}

\subsection{Proof of Lemma \ref{lemma:fixpoint}}

\begin{proof}
Recall that, given $s\in \Qrelkleisli{\Rsemiring}(Y+X,X)$ $\fix_{Y,X} s=\sup_n s^{n}$, where $s^{0}=0$ and $s^{n+1}=s\circ\langle\pi_1, s^{n}\rangle$. We will prove, by induction on $n$, that 
$\alpha\circ \langle t^n,1_{\C P}\rangle=\tilde{t}^n$, from which the claim will follow by continuity. 
Using the two hypotheses of the Lemma we have then
\begin{align*}
\alpha\circ \langle t^0,1_{\C P}\rangle=\alpha\circ \langle 0,1_{\C P}\rangle =0=s^0
\end{align*}
and
\begin{align*}
\alpha\circ \langle t^{n+1},1_{\C P}\rangle&=
\alpha\circ \langle t\circ t^{n},1_{\C P}\rangle\\
&=
\alpha\circ \langle t\circ\pi_1,\pi_2\rangle
\circ 
\langle t^{n},1_{\C P}\rangle\\
&=
\tilde t\circ \langle\alpha\circ \langle\pi_1,\pi_2\rangle,\pi_2 \rangle\circ 
\langle t^{n},1_{\C P}\rangle\\
&=
\tilde t\circ \langle\alpha\circ \langle t^n,1_{\C P}\rangle,1_{\C P}\rangle\\
&\stackrel{\tiny\text{[I.H.]}}{=}
\tilde t\circ \langle\tilde t^n,1_{\C P}\rangle=\tilde t^{n+1}.
\end{align*}
\end{proof}

\end{document}